\documentclass{article}

     \PassOptionsToPackage{numbers, compress}{natbib}

     \usepackage[final]{neurips_2020}

\usepackage[utf8]{inputenc} 
\usepackage[T1]{fontenc}    
\usepackage{hyperref}       
\usepackage{url}            
\usepackage{booktabs}       
\usepackage{amsfonts}       
\usepackage{nicefrac}       
\usepackage{microtype}      

\usepackage{graphicx}
\usepackage{subfigure}
\usepackage{amssymb}
\usepackage{amsthm}
\usepackage{amsmath}
\usepackage{multirow}
\usepackage{graphics}

\usepackage{algorithm}
\usepackage{algorithmic}
\usepackage{color, colortbl}

\definecolor{Gray}{gray}{0.9}

\usepackage[table,xcdraw]{xcolor}

\usepackage{multirow}

\usepackage{diagbox}

\newtheorem*{theorem*}{Theorem}
\newtheorem{theorem}{Theorem}

\newtheorem{corollary}{Corollary}
\newtheorem{corollary*}{Corollary}

\title{Subgroup-based Rank-1 Lattice Quasi-Monte Carlo}

\author{%
  Yueming Lyu   \\
  Australian Artificial Intelligence Institute\\
  University of Technology Sydney\\
  \texttt{yueminglyu@gmail.com} \\
   \And
   Yuan Yuan \\
   CSAIL \\
   Massachusetts Institute of Technology \\
   \texttt{miayuan@mit.edu} \\
   \AND
   Ivor W. Tsang\ \\
  Australian Artificial Intelligence Institute\\
  University of Technology Sydney\\
   \texttt{ Ivor.Tsang@uts.edu.au} \\
}

\begin{document}

\maketitle

\begin{abstract}
 Quasi-Monte Carlo (QMC) is an essential tool for integral approximation, Bayesian inference, and sampling for simulation in science, etc. In the QMC area, the rank-1 lattice is important due to its simple operation, and nice properties for point set construction.  However, the construction of the  generating vector of the rank-1 lattice is usually time-consuming because of an exhaustive computer search.  To address this issue,  we propose a simple closed-form rank-1 lattice construction method based on group theory.  Our method reduces the number of distinct pairwise distance values to generate a more regular lattice. We theoretically prove a lower and an upper bound of the minimum pairwise distance of any non-degenerate rank-1 lattice. Empirically, our methods can generate a near-optimal rank-1 lattice compared with the Korobov exhaustive search regarding the $l_1$-norm and $l_2$-norm minimum distance. Moreover, experimental results show that our method achieves superior approximation performance on  benchmark integration test problems and  kernel approximation problems. 
\end{abstract}

\section{Introduction}
\label{Introduction}

Integral operation is critical in a large amount of interesting machine learning applications, e.g. kernel approximation with random feature maps~\cite{rahimi2008random}, variational  inference in Bayesian learning~\cite{beal2003variational},    generative modeling and variational autoencoders~\cite{kingma2013auto}. Directly calculating an integral is usually infeasible in these real applications. Instead, researchers usually try to find an approximation for the integral.    A simple and conventional approximation  is  Monte Carlo (MC) sampling, in  which the integral is approximated by calculating the average of the i.i.d.  sampled integrand values. 
Monte Carlo (MC) methods~\cite{hammersley2013monte} are  widely studied with many techniques to reduce the approximation error, which includes  importance sampling and variance reduction techniques and more~\cite{arouna2004adaptative}. 



To further reduce the approximation error, Quasi-Monte Carlo (QMC) methods utilize a low discrepancy point set instead of the i.i.d. sampled point set used in the standard Monte Carlo method. There are two main research lines in the area of QMC~\cite{dick2013high,niederreiter1992random}, i.e., the digital nets/sequences and  lattice rules. The Halton sequence and the Sobol sequence are the widely used representatives of digital sequences~\cite{dick2013high}. Compared with digital nets/sequences, the points set of lattice rules preserve the properties of lattice. The points partition the space into small repeating cells.  Among  previous research on the lattice rules, Korobov introduced integration lattice rules  in~\cite{korobov1959approximate} for an integral approximation of the periodic integrands. \cite{sloan2001tractability} proves that there also exist good lattice rules for non-periodic integrands. According to  general lattice rules, a point set is usually constructed by enumerating the integer vectors and multiplying them with an invertible generator matrix. A general lattice rule has to check each constructed point to see whether it is inside a unit cube and discard it if it is not. The process is  repeated until we reach the desired number of points. This construction process is inefficient since the checking step is required for every point.  Note that rescaling the unchecked points    by the maximum norm of all the  points may lead to non-uniform points set in the cube.

An interesting special case of the lattice rules is the rank-1 lattice, which only requires one generating vector to construct the whole point set. Given the generating vector,  rank-1 lattices can be obtained by a very simple construction form. It is thus much more efficient to construct the point set with the simple construction form. Compared with the general lattice rule, the construction form of the rank-1 lattice has already guaranteed the constructed point to be inside the unit cube, therefore, no further checks are required. We refer to \cite{dick2013high} and \cite{niederreiter1992random} for a more detailed survey of QMC and rank-1 lattice.

\begin{figure*}[t]
\label{demo}
\centering
\subfigure[\scriptsize{ i.i.d.  Monte Carlo sampling}]{
\label{fig2c_2}
\includegraphics[width=0.3\linewidth]{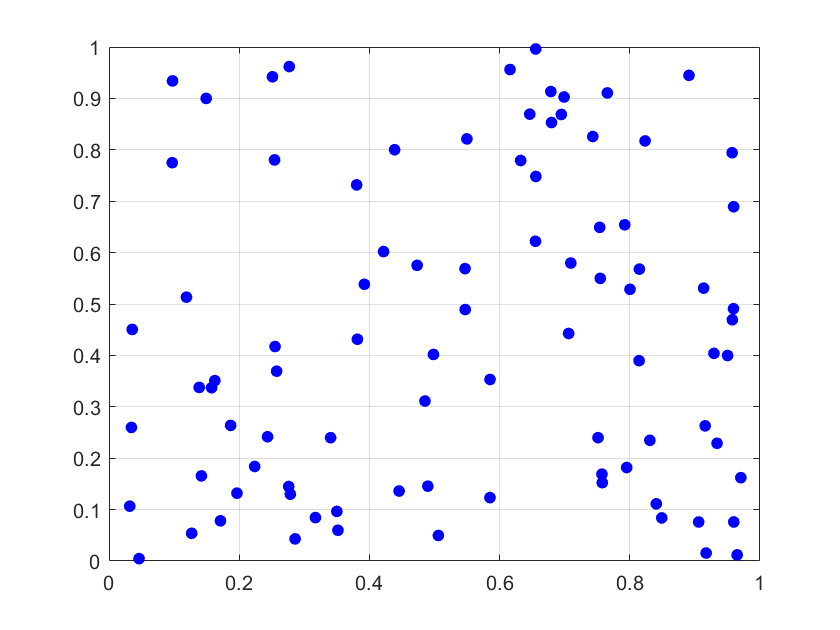}}
\subfigure[\scriptsize{ Sobol sequence}]{
\label{fig2c_3}
\includegraphics[width=0.3\linewidth]{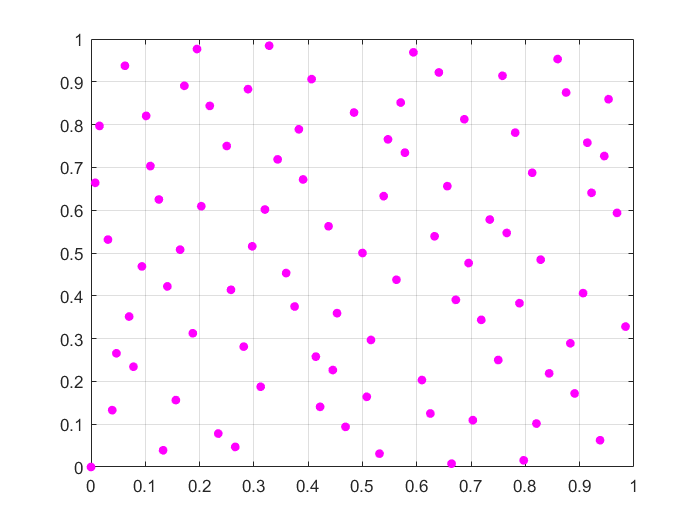}}
\subfigure[\scriptsize{Our subgroup rank-1 lattice }]{
\label{fig2a_l}
\includegraphics[width=0.3\linewidth]{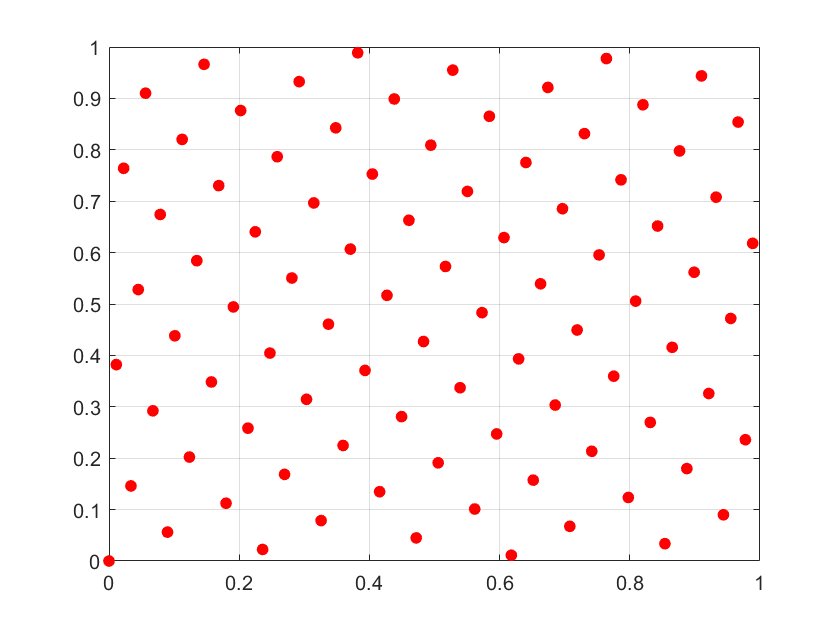}}
\caption{The 89  points constructed by i.i.d. Monte Carlo sampling, Sobol sequence and our subgroup rank-1 lattice  on $[0,1]^2$.}
\end{figure*}

Although the rank-1 lattice can derive a simple construction form,  obtaining the generating vector remains difficult. Most methods~\cite{korobov1960properties,nuyens2006fast,doerr2013constructing,leobacher2014introduction,l2016algorithm,laimer2017combined,owen2019monte} in the literature rely on an exhaustive computer search by optimizing some criteria  to find a good generating vector. Korobov~\cite{korobov1960properties} suggests searching the generating vector in a form of $ [1,\alpha,\alpha^2,\cdots,\alpha^{d-1}]$ with $\alpha \in \{1,\cdots,n-1\}$, where $ d $ is the dimension and $ n $ is the number of points, such that the greatest common divisor of $ \alpha $ and $ n $ equals to 1. Sloan et al. study the component-by-component construction for the lattice rules~\cite{sloan2002component}. It is a greedy search that is faster than an exhaustive search.
Nuyens et al.~\cite{nuyens2006fast} propose a fast algorithm to construct the generating vector using a component-by-component search method. Although the exhaustive checking steps are avoided compared with general lattice rules, the rank-1 lattice still requires a brute-force search for the generating vector, which is still very time-consuming, especially when the dimension and the number of points are large.

To address this issue, we propose a closed-form rank-1 lattice rule that directly computes a generating vector without any search process. To generate a more evenly spaced lattice, we propose to reduce the number of distinct pairwise distance in the lattice point set to make the lattice more regular w.r.t. the minimum toroidal distance~\cite{minTdistance}. Larger minimum toroidal distance means more regular. 
Based on group theory, we derive that if the generating vector $\boldsymbol{z}$ satisfies the condition that set $\{\boldsymbol{z},-\boldsymbol{z}\}:=\{z_1,\cdots, z_d,-z_1,\cdots,-z_d\}$ is a subgroup of the multiplicative group of integers modulo $n$, where $n$ is the number of points, then the number of distinct pairwise distance can be efficiently reduced.
We construct the generating vector by ensuring this condition.
With the proposed subgroup-based rank-1 lattice, we can construct a more evenly spaced lattice. An illustration of the  generated lattice is shown in Figure~\ref{demo}.

Our contributions are summarized as follows:
\begin{itemize}
    \item We propose a simple and efficient closed-form method for rank-1 lattice construction, which does not require the time-consuming exhaustive computer search that previous rank-1 lattice algorithms rely on. A side product is a closed-form method to generate QMC points set on sphere $\mathbb{S}^{d-1}$ with bounded mutual coherence, which is presented in Appendix.

   \item We generate a more regular lattice by reducing the number of distinct pairwise distances. We prove a lower and an upper bound for the minimum $l_1$-norm-based and $l_2$-norm-based toroidal distance of the rank-1 lattice. Theoretically, our constructed lattice is the optimal rank-1 lattice for maximizing the minimum toroidal distance when the number of points $n$ is a prime number and $n=2d+1$. 

   \item Empirically, the proposed method generates near-optimal rank-1 lattice compared with the Korobov search method in maximizing the minimum of the $l_1$-norm-based and $l_2$-norm-based toroidal distance.
   
    \item Our method obtains better approximation accuracy on  benchmark test problems and  kernel approximation problem. 
\end{itemize}

\section{Background}

We first give the definition and the properties of lattices in Section~\ref{lattice}. Then we introduce the minimum distance criterion for lattice construction in Section~\ref{background_distance}.

\subsection{The Lattice}
\label{lattice}

A $d$-dimensional lattice $\Lambda$ is a set of points that contains no limit points and satisfies~\cite{lyness2003notes}
\begin{align}
   \forall \boldsymbol{x},\boldsymbol{x}' \in \Lambda \Rightarrow \boldsymbol{x}+ \boldsymbol{x}' \in \Lambda \; \text{and} \; \boldsymbol{x} - \boldsymbol{x}' \in \Lambda. 
\end{align}
A widely known lattice is the unit lattice $\mathbb{Z}^d$ whose components are all integers. A general lattice is  constructed by a generator matrix. Given a generator matrix $\boldsymbol{B} \in \mathbb{R}^{d \times d}$, a $d$-dimensional lattice $\Lambda$ can be constructed as 
\begin{align}
\label{MatrixL}
 \Lambda =    \{\boldsymbol{By} \big | \; \boldsymbol{y} \in \mathbb{Z}^d  \}.
\end{align}
A generator matrix is not unique to a lattice $\Lambda$, namely, a lattice $\Lambda$ can be obtained from a different generator matrices.

A lattice point set for integration is constructed as $\Lambda \cap [0,1)^d$. This step may require an additional search (or check) for all the points inside the unit cube.

A rank-1 lattice is a special case of the general lattice, which has a simple operation for point set construction instead of directly using Eq.(\ref{MatrixL}).
A rank-1 lattice point set  can be constructed as
 \begin{align}
 \label{rank1lattice}
  \boldsymbol{x}_i :=  \left\langle   \frac{i \boldsymbol{z}  }{n}  \right\rangle , i \in \{ 0,1,\cdots,n-1 \},
 \end{align}
 where $\boldsymbol{z} \in \mathbb{
 Z}^d $ is the so-called generating vector, and the big $ \langle\cdot\rangle $ denotes the operation of taking  the fractional part of the input number elementwise.  Compared with the general lattice rule, the construction form of the rank-1 lattice already  ensures the constructed points to be inside the unit cube without the need for any further checks.
 
 Given a rank-1 lattice set $X$ in the unit cube, we can also construct a randomized point set. Sample a random variable $\boldsymbol{\Delta} \sim Uniform [0,1]^d$, we can construct a   point set $\widetilde{X}$ by random shift as~\cite{dick2013high}
\begin{align}
\label{randomshift}
    \widetilde{X} = \left\langle X + \boldsymbol{\Delta}    \right\rangle.
\end{align}

\subsection{The separating distance of a lattice}
\label{background_distance}

Several criteria have been studied in the literature for good lattice construction through computer search. Worst case error is one of the most widely used criteria for functions in a reproducing kernel Hilbert space (RKHS)~\cite{dick2013high}. However, this criterion requires the prior knowledge of functions and the assumption of the RKHS. Without assumptions of the functions,  it is reasonable to construct a good lattice by designing an evenly spaced point set.  Minimizing the covering radius is  a good way for evenly spaced point set construction.

As minimizing the covering radius of the lattice is equivalent to maximizing the packing radius~\cite{pack}, we can construct the  point set through maximizing the packing radius (separating distance) of the lattice.  Define  the covering radius  and  packing radius  of a set of points $X=\{x_1,...,x_N\}$ as Eq.(\ref{hx}) and Eq.(\ref{rx}), respectively:
\begin{align}\label{hx}
     h_X =  \sup_{x\in \mathcal{X}}{\min _{x_k \in X }} \|x-x_k \|,  \\
     \rho_X = \frac{1}{2} \min _ {\substack{x_i,x_j  \in X, \\ x_i \ne x_j}} \|x_i -x_j \| . \label{rx}
\end{align}
The $l_p$-norm-based toroidal distance \cite{minTdistance} between two lattice points $ {\bf{x}}\in [0,1]^d $ and $ {\bf{x}'}\in [0,1]^d $ can be defined as: 
 \begin{equation}
\|{\bf{x}} - {\bf{x}'}\|_{T_p} := \!\!\left( \sum_{i=1}^d (\min( |x_i - x'_i|  , 1-|x_i - x'_i|))^p  \right)^{\frac{1}{p}}
\end{equation}
Because the difference (subtraction) between two lattice points is still a lattice point, and a rank-1 lattice has a period 1, the packing radius  $ \rho_X $ of a rank-1 lattice can be calculated as
\begin{align}\label{rank1D}
    \rho_X = \min_{{\bf{x}} \in {X \setminus  {\bf{0}}}}\frac{1}{2}\|{\bf{x}} \|_{T_2},
\end{align}
where $ \|{\bf{x}} \|_{T_2} $ denotes the $l_2$-norm-based toroidal distance between $ \bf{x} $ and $ \bf{0} $, symbol $X \setminus  {\bf{0}}$ denotes the set $X$ excludes the point $\boldsymbol{0}$.
This formulation calculates the packing radius  with a time complexity of $ \mathcal{O}(Nd) $ rather than $ \mathcal{O}(N^2d) $ in pairwise comparison. 
However, the computation of the packing radius is not easily accelerated by fast Fourier transform due to the minimum operation in Eq.(\ref{rank1D}).


\section{Subgroup-based Rank-1 Lattice }

In this section, we derive our construction of a rank-1 lattice based on the subgroup theory. Then we analyze the properties of our method. We provide detailed proofs in the supplement.

\subsection{Construction of the Generating Vector}
From the definition of rank-1 lattice, we know the packing radius of rank-1 lattice with $n$ points can be reformulated as 
\begin{align}\label{rank1D2}
    \rho_X = \min_{  i \in \{1,\cdots,n-1\} }\frac{1}{2}\|{\bf{x}}_i \|_{T_2},
\end{align}
where
 \begin{align}
  \boldsymbol{x}_i := \frac{i \boldsymbol{z}  \; \text{mod} \;n}{n}, i \in \{1,...,n-1\}.
 \end{align}
 Then, we have 
 \begin{align}
     \rho_X & = \min_{  i \in \{1,\cdots,n-1\} }\frac{1}{2} \left\| \text{min} \left( \frac{i \boldsymbol{z}  \; \text{mod} \;n}{n},  \frac{ n- i \boldsymbol{z}  \; \text{mod} \;n}{n} \right)  \right\|_{2} \nonumber
   \\  & =  \min_{  i \in \{1,\cdots,n-1\} }\frac{1}{2} \left\| \text{min} \left( \frac{i \boldsymbol{z}  \; \text{mod} \;n}{n},  \frac{  (-i\boldsymbol{z})   \; \text{mod} \;n}{n} \right)  \right\|_{2}, 
 \end{align}
 where $\text{min}(\cdot,\cdot)$ denotes the elementwise min operation between two inputs.
 
Suppose $n$ is a prime number, from  number theory, we know that for a primitive root $g$, the residue of $\{g^0,g^1,\cdots,g^{n-2} \} $  modulo $n$ forms a cyclic group under multiplication, and $g^{n-1} \equiv 1 \; \text{mod} \;n$.  Since $ (g^{\frac{n-1}{2}})^2 =g^{n-1} \equiv 1 \; \text{mod} \;n$, we know  that $g^{\frac{n-1}{2}} \equiv -1 \; \text{mod} \;n$.
 
 Because of the one-to-one correspondence between the residue of $\{g^0,g^1,\cdots,g^{n-2} \} $  modulo $n$ and the set $\{1,2,\cdots,n-1\}$, we can construct the generating vector as
 \begin{align}
     \boldsymbol{z} = [ g^{m_1},g^{m_2},\cdots,g^{m_d} ] \; \text{mod} \; n
 \end{align}
 without loss of generality,  where
$m_1,\cdots,m_d$ are  integer components to be designed. Denote 
$ \boldsymbol{\Bar{z}} = [ g^{\frac{n-1}{2}+m_1},g^{\frac{n-1}{2}+m_2},\cdots,g^{\frac{n-1}{2}+m_d} ] \; \text{mod} \; n$, maximizing the separating distance $\rho_X$ is equivalent to maximizing 
\begin{align}
    J \!  = \! \min_{  k \in \{0,\cdots,n-2\} } \left\| \text{min} (g^k\boldsymbol{z} \; \text{mod} \; n ,g^k\boldsymbol{\Bar{z}} \; \text{mod} \; n)  \right\|_2. 
\end{align}
Suppose $2d$ divides $n-1$, i.e., $2d|(n-1)$, by setting $m_i=g^{\frac{(i-1)(n-1)}{2d}}$ for $i \in \{1,\cdots,d\}$, we know that $H =\{g^{m_1},g^{m_2},\cdots,g^{m_d}, g^{\frac{n-1}{2}+m_1},g^{\frac{n-1}{2}+m_2},\cdots,g^{\frac{n-1}{2}+m_d}  \}$ is equivalent to  setting $\{g^0,g^{\frac{n-1}{2d}},\cdots, g^{\frac{(2d-1)(n-1)}{2d} } \}$ mod $n$ , and it forms a subgroup of the group $\{g^0,g^1,\cdots,g^{n-2}\}$ mod $n$.
From Lagrange's theorem in group theory~\cite{dummit2004abstract},  we know that the cosets of the subgroup $H$ partition the entire group $\{g^0,g^1,\cdots,g^{n-2}\}$ into equal-size, non-overlapping sets, and the number of cosets of $H$ is $\frac{n-1}{2d}$. Thus, we know that distance $\text{min} (g^k\boldsymbol{z} \; \text{mod} \; n ,g^k\boldsymbol{\Bar{z}} \; \text{mod} \; n) $ for $k \in \{0,\cdots,n-2\} $ has  $\frac{n-1}{2d}$ different values, and there are the same numbers of items for each value.

Thus, we can construct the generating vector as 
\begin{align}
     \boldsymbol{z} = [ g^0,g^{\frac{n-1}{2d}}, g^{\frac{2(n-1)}{2d}},\cdots, g^{\frac{(d-1)(n-1)}{2d} } ] \; \text{mod} \; n. 
\end{align}
In this way,  the constructed rank-1 lattice is more regular as it has few  different distinct pairwise distance values, and for each distance, the same number of items obtain this value.  Usually, the constructed regular lattice is more evenly spaced,  and it has a large minimum pairwise distance. We confirm this empirically through extensive experiments in Section~\ref{experiments}.

We summarize our construction method and the properties of the constructed rank-1 lattice in Theorem~\ref{construction}.

\begin{theorem}
\label{construction}
Suppose $n$ is a prime number and $2d | (n-1)$. Let $g$ be a  primitive root of $n$. Let $  \boldsymbol{z} = [ g^0,g^{\frac{n-1}{2d}}, g^{\frac{2(n-1)}{2d}},\cdots, g^{\frac{(d-1)(n-1)}{2d} } ] \; \text{mod} \; n $. Construct a rank-1 lattice $X \!=\!\{\boldsymbol{x}_0,\cdots,\boldsymbol{x}_{n-1}\}$ with $\boldsymbol{x}_i = \frac{i \boldsymbol{z}  \; \text{mod} \;n}{n}, i \in \{0,...,n-1\}  $. Then, there are $\frac{n-1}{2d}$ distinct pairwise  toroidal distance values among $X$, and  each distance value is taken by the  same number of pairs in $X$.
\end{theorem}

As shown in Theorem~\ref{construction}, our method can construct regular rank-1 lattice through a very simple closed-form construction, which does not require any exhaustive computer search.

\subsection{Regular Property of Rank-1 Lattice}

We show the regular property of rank-1 lattices in terms of $l_p$-norm-based toroidal distance.

\begin{theorem}
\label{Bounds}
Suppose $n$ is a prime number and $n \ge 2d+1$.  Let $  \boldsymbol{z} = [ z_1,z_2, \cdots, z_d  ] $ with $1 \le z_k \le n-1 $.  Construct a rank-1 lattice $X=\{\boldsymbol{x}_0,\cdots,\boldsymbol{x}_{n-1}\}$ with $\boldsymbol{x}_i = \frac{i \boldsymbol{z}  \; \text{mod} \;n}{n}, i \in \{0,...,n-1\}  $ and $z_i \ne z_j$ . Then, the minimum pairwise toroidal distance can be bounded as 
\begin{align}
&  \frac{d(d+1)}{2n}  \le \min_{  i,j \in \{0,\cdots,n-1\} ,  i \ne j} \|{\bf{x}}_i - {\bf{x}}_j \|_{T_1} \le \frac{(n+1)d}{4n}  \\
 &   \frac{\sqrt{6d(d+1)(2d+1)}}{6n}  \le \min_{  i,j \in \{0,\cdots,n-1\} ,  i \ne j} \|{\bf{x}}_i - {\bf{x}}_j \|_{T_2}\le  \sqrt{\frac{(n+1)d}{12n} },
\end{align}
where $\|\cdot\|_{T_1}$ and  $\|\cdot\|_{T_2}$ denote the $l_1$-norm-based toroidal distance and the $l_2$-norm-based toroidal distance, respectively.
\end{theorem}
Theorem~\ref{Bounds} gives the upper and lower bounds of the minimum pairwise distance of any  non-degenerate rank-1 lattice.  The term `non-degenerate' means that  the elements   in the generating vector are not equal, i.e., $z_i \ne z_j$.

We now show that our subgroup-based rank-1 lattice can achieve the optimal minimum pairwise distance when $n=2d+1$ is a prime number.

\begin{corollary}
\label{optimalCor}
Suppose $n=2d+1$ is a prime number. Let $g$ be a  primitive root of $n$. Let $  \boldsymbol{z} = [ g^0,g^{\frac{n-1}{2d}}, g^{\frac{2(n-1)}{2d}},\cdots, g^{\frac{(d-1)(n-1)}{2d} } ] \; \text{mod} \; n $. Construct rank-1 lattice $X=\{\boldsymbol{x}_0,\cdots,\boldsymbol{x}_{n-1}\}$ with $\boldsymbol{x}_i = \frac{i \boldsymbol{z}  \; \text{mod} \;n}{n}, i \in \{0,...,n-1\}  $. Then,  the pairwise toroidal distance of the lattice $X$  attains the upper bound. 
\begin{align}
    &  \|{\bf{x}}_i - {\bf{x}}_j \|_{T_1} = \frac{(n+1)d}{4n} ,  \forall i,j \in \{0,\cdots,n-1\} ,  i \ne j,  \\
 &  \|{\bf{x}}_i - {\bf{x}}_j \|_{T_2} = \sqrt{\frac{(n+1)d}{12n} } ,  \forall i,j \in \{0,\cdots,n-1\} ,  i \ne j . 
\end{align}

\end{corollary}

Corollary~\ref{optimalCor} shows a case when our subgroup rank-1 lattice obtains the maximum minimum pairwise toroidal distance. It is useful for expensive black-box functions, where the number of function queries is small. Empirically, we find that our subgroup rank-1 lattice can achieve near-optimal pairwise toroidal distance in many other cases.

\section{QMC for Kernel Approximation}

Another application of our subgroup rank-1 lattice is kernel approximation. Kernel approximation has been widely studied. A random feature maps is a promising way for kernel approximation.   Rahimi et al. study the shift-invariant kernels by Monte Carlo sampling~\cite{rahimi2008random}.  Yang et al. suggest employing QMC for kernel approximation~\cite{yang2014quasi,avron2016quasi}. Several previous methods work on the construction of structured feature maps for kernel approximation~\cite{le2013fastfood,choromanski2016recycling,lyu2017spherical}. Apart from other kernel approximation methods designed for specific kernels, QMC can serve as a plug-in for any integral representation of kernels to improve kernel approximation. We include this section to be self-contained.

From Bochner's Theorem, shift invariant kernels can be expressed as an  integral~\cite{rahimi2008random} 
\begin{equation}
\label{K}
\begin{array}{l}
{\rm{K}}(\boldsymbol{x},\boldsymbol{y}) =  \int_{{{\rm{\mathbb{R}}}^d}} {{e^{ - i{{(\boldsymbol{x} - \boldsymbol{y})}^\top}\bf{w}}}p(\bf{w})d\bf{w}},
\end{array}
\end{equation}
where $i = \sqrt{-1}$, and $p(\bf{w})$ is a probability density. $ p({\bf{w}}) = p(- {\bf{w}}) \geq 0 $   ensure the imaginary parts of the integral vanish.
Eq.(\ref{K}) can be rewritten as 
\begin{align}
{\rm{K}}(\boldsymbol{x},\boldsymbol{y}) = \int _{[0,1]^d}  {{e^{ - i{{(\boldsymbol{x} - \boldsymbol{y})}^\top}\Phi^{-1}(\boldsymbol{\epsilon})}} d\boldsymbol{\epsilon}}.
\end{align}
We can approximate the integral~Eq.(\ref{K}) by using our subgroup rank-1 lattice according to the QMC approximation in \cite{yang2014quasi,tompkins2019black}
\begin{align}
{\rm{K}}(\boldsymbol{x},\boldsymbol{y}) & = \int _{[0,1]^d}  {{e^{ - i{{(\boldsymbol{x} - \boldsymbol{y})}^\top}\Phi^{-1}(\boldsymbol{\epsilon})}} d\boldsymbol{\epsilon}}  \approx \frac{1}{n} \sum  _{i=1} ^{n} {e^{ - i{{(\boldsymbol{x} - \boldsymbol{y})}^\top}\Phi^{-1}(\boldsymbol{\epsilon}_i)}}  = \left< \Psi(\boldsymbol{x}), \Psi(\boldsymbol{y}) \right>,
\end{align}
where $\Psi(\boldsymbol{x}) =\frac{1}{\sqrt{n}} \left[e^{-i\boldsymbol{x}^\top \Phi^{-1}(\boldsymbol{\epsilon}_1) },\cdots,e^{-i\boldsymbol{x}^\top \Phi^{-1}(\boldsymbol{\epsilon}_n) } \right]$ is the feature map of  the input  $\boldsymbol{x}$.

\vspace{-2mm}
\section{Experiments}
\label{experiments}
\vspace{-1mm}

In this section, we first evaluate the minimum distance generated by our subgroup rank-1 lattice in section~\ref{distance}. We then evaluate the subgroup rank-1 lattice on integral approximation tasks and kernel approximation task in section~\ref{testInapp} and~\ref{kerApp}, respectively. 



\begin{table}[t]
\centering
\begin{footnotesize}
\caption{Minimum  $l_1$-norm-based toroidal distance of rank-1 lattice constructed by different methods.}
\label{tab:l1normD}
\resizebox{\columnwidth}{!}{
\begin{tabular}{ccccccccccccc}
\toprule
\multirow{4}{*}{\begin{tabular}[c]{@{}l@{}} d=50 \end{tabular}}
&  & n=101             & 401             & 601             & 701             & 1201    & 1301            & 1601            & 1801            & 1901            & 2801         \\ 
& SubGroup & \textbf{12.624} & \textbf{11.419} & \textbf{11.371} & \textbf{11.354} & \textbf{11.029}  & \textbf{10.988} & 10.541          & 10.501          & 10.454          & \textbf{10.748} \\ 
& Hua~\cite{hua2012applications}      & 10.426          & 10.421          & 9.8120          & 10.267          & 10.074        & 9.3982          & 9.5890          & 9.5175          & 8.9868          & 9.2260      \\ 
& Korobov~\cite{korobov1960properties}  & \textbf{12.624} & \textbf{11.419} & \textbf{11.371} & \textbf{11.354} & \textbf{11.029} & \textbf{10.988} & \textbf{10.665} & \textbf{10.561} & \textbf{10.701} & \textbf{10.748} \\ 
                            
\midrule

\multirow{4}{*}{\begin{tabular}[c]{@{}l@{}} d=100 \end{tabular}}
          &        &401             &601             &1201            &1601            &1801        &2801            &3001            &4001            &4201            &4801         \\ 
& SubGroup & \textbf{24.097} & \textbf{23.760} & 22.887          & \textbf{23.342} & 22.711       & \textbf{23.324} & 22.233          & \textbf{22.437} & 22.573          & 21.190      \\ 
& Hua~\cite{hua2012applications}      & 21.050          & 21.251          & 21.205          & 20.675          & 19.857       & 20.683          & 20.700          & 19.920          & 19.967          & 20.574       \\ 
& Korobov~\cite{korobov1960properties}  & \textbf{24.097} & \textbf{23.760} & \textbf{23.167} & \textbf{23.342} & \textbf{22.893}  & \textbf{23.324} & \textbf{22.464} & \textbf{22.437} & \textbf{22.573} & \textbf{22.188}  \\ 

\midrule

\multirow{4}{*}{\begin{tabular}[c]{@{}l@{}} d=200 \end{tabular}}
          &     &401             & 1201            & 1601            & 2801            & 4001      & 4801            & 9601            & 12401           & 14401           & 15601          \\ 
& SubGroup & \textbf{50.125} & \textbf{48.712} & 47.500          & 47.075          & \textbf{47.810} & 45.957          & 45.819          & \textbf{46.223} & 43.982          & \textbf{45.936} \\ 
& Hua~\cite{hua2012applications}      & 43.062          & 43.057          & 43.052          & 43.055          & 43.053   & 43.055          & 43.053          & 42.589          & 42.558          & 42.312        \\ 
& Korobov~\cite{korobov1960properties}  & \textbf{50.125} & \textbf{48.712} & \textbf{47.660} & \textbf{47.246} & \textbf{47.810}  & \textbf{46.686} & \textbf{46.154} & \textbf{46.223} & \textbf{45.949} & \textbf{45.936} \\ 

\midrule

\multirow{4}{*}{\begin{tabular}[c]{@{}l@{}} d=500 \end{tabular}}
      &     & 3001            & 4001            & 7001            & 9001            & 13001       & 16001           & 19001           & 21001           & 24001           & 28001         \\  
& SubGroup & \textbf{121.90} & \textbf{121.99} & 119.60          & 118.63          & \textbf{120.23} & \textbf{119.97} & 116.41          & \textbf{120.56} & \textbf{120.24} & 113.96   \\ 
& Hua~\cite{hua2012applications}      & 108.33          & 108.33          & 108.33          & 108.33          & 108.33      & 108.33          & 108.33          & 108.33          & 108.33          & 108.33       \\ 
& Korobov~\cite{korobov1960properties}  & \textbf{121.90} & \textbf{121.99} & \textbf{120.46} & \textbf{120.16} & \textbf{120.23}  & \textbf{119.97} & \textbf{119.41} & \textbf{120.56} & \textbf{120.24} & \textbf{118.86} \\ 

 \bottomrule
\end{tabular}
}
\end{footnotesize}
\end{table}
\vspace{-2mm}


\begin{table}[t]
\centering
\begin{footnotesize}
\caption{Minimum  $l_2$-norm-based toroidal distance of rank-1 lattice constructed by different methods.}
\label{tab:l2normD}
\resizebox{\columnwidth}{!}{
\begin{tabular}{ccccccccccccc}
\toprule
\multirow{4}{*}{\begin{tabular}[c]{@{}l@{}} d=50 \end{tabular}}
                          &          &  n=101             &  401             &  601             &  701             &  1201     &  1301            &  1601            &  1801            &  1901            &  2801           \\ 
& SubGroup & \textbf{2.0513} & \textbf{1.9075} & \textbf{1.9469} & \textbf{1.9196} & \textbf{1.8754} & 1.8019          & 1.8008          & \textbf{1.8709} & 1.7844          & 1.7603  \\  
& Hua~\cite{hua2012applications}      & 1.7862          & 1.7512          & 1.7293          & 1.7049          & 1.7326    & 1.6295          & 1.6659          & 1.6040          & 1.5629          & 1.5990       \\  
& Korobov~\cite{korobov1960properties}  & \textbf{2.0513} & \textbf{1.9075} & \textbf{1.9469} & \textbf{1.9196} & \textbf{1.8754}  & \textbf{1.8390} & \textbf{1.8356} & \textbf{1.8709} & \textbf{1.8171} & \textbf{1.8327}  \\ 
                            
\midrule

\multirow{4}{*}{\begin{tabular}[c]{@{}l@{}} d=100 \end{tabular}}
                  &     & 401             &  601             & 1201            & 1601            & 1801      & 2801            & 3001            & 4001            & 4201            & 4801         \\ 
& SubGroup & \textbf{2.8342} & \textbf{2.8143} & 2.7077          & \textbf{2.7645} & \textbf{2.7514} & 2.6497          & 2.6337          & 2.6410          & 2.6195          & 2.5678   \\ 
& Hua~\cite{hua2012applications}      & 2.5385          & 2.5739          & 2.4965          & 2.4783          & 2.4132       & 2.5019          & 2.4720          & 2.4138          & 2.4537          & 2.4937     \\ 
& Korobov~\cite{korobov1960properties}  & \textbf{2.8342} & \textbf{2.8143} & \textbf{2.7409} & \textbf{2.7645} & \textbf{2.7514} & \textbf{2.6956} & \textbf{2.6709} & \textbf{2.6562} & \textbf{2.6667} & \textbf{2.6858}  \\ 

\midrule

\multirow{4}{*}{\begin{tabular}[c]{@{}l@{}} d=200 \end{tabular}}
               &     &401             &1201            & 1601            &2801            & 4001     &4801            &9601            &12401           &14401           &15601         \\ 
& SubGroup & \textbf{4.0876} & \textbf{3.9717} & \textbf{3.9791} & 3.8425          & \textbf{3.9276} & 3.8035          & 3.7822          & \textbf{3.8687} & 3.6952          & 3.8370    \\ 
& Hua~\cite{hua2012applications}      & 3.7332          & 3.7025          & 3.6902          & 3.6944          & 3.7148     & 3.6936          & 3.6571          & 3.5625          & 3.6259          & 3.5996          \\ 
& Korobov~\cite{korobov1960properties}  & \textbf{4.0876} & \textbf{3.9717} & \textbf{3.9791} & \textbf{3.9281} & \textbf{3.9276} & \textbf{3.9074} & \textbf{3.8561} & \textbf{3.8687} & \textbf{3.8388} & \textbf{3.8405} \\ 

\midrule

\multirow{4}{*}{\begin{tabular}[c]{@{}l@{}} d=500 \end{tabular}}
            &     &3001            &4001            &7001            &9001            &13001      &16001           &19001           &21001           &24001           &28001        \\ 
& SubGroup & \textbf{6.3359} & \textbf{6.3769} & 6.3141          & 6.2131          & \textbf{6.2848} & 6.2535          & 6.0656          & \textbf{6.2386} & \textbf{6.2673} & 6.1632  \\ 
& Hua~\cite{hua2012applications}      & 5.9216          & 5.9216          & 5.9215          & 5.9215          & 5.9216      & 5.9216          & 5.9215          & 5.9215          & 5.8853          & 5.9038           \\ 
& Korobov~\cite{korobov1960properties}  & \textbf{6.3359} & \textbf{6.3769} & \textbf{6.3146} & \textbf{6.2960} & \textbf{6.2848}   & \textbf{6.2549} & \textbf{6.2611} & \textbf{6.2386} & \textbf{6.2673} & \textbf{6.2422} \\  

 \bottomrule
\end{tabular}
}
\end{footnotesize}
\end{table}
\vspace{-2mm}

\begin{figure}[t]
\centering
\subfigure[\scriptsize{ 50-d Integral Approximation }]{
\label{All50}
\includegraphics[width=0.224\linewidth]{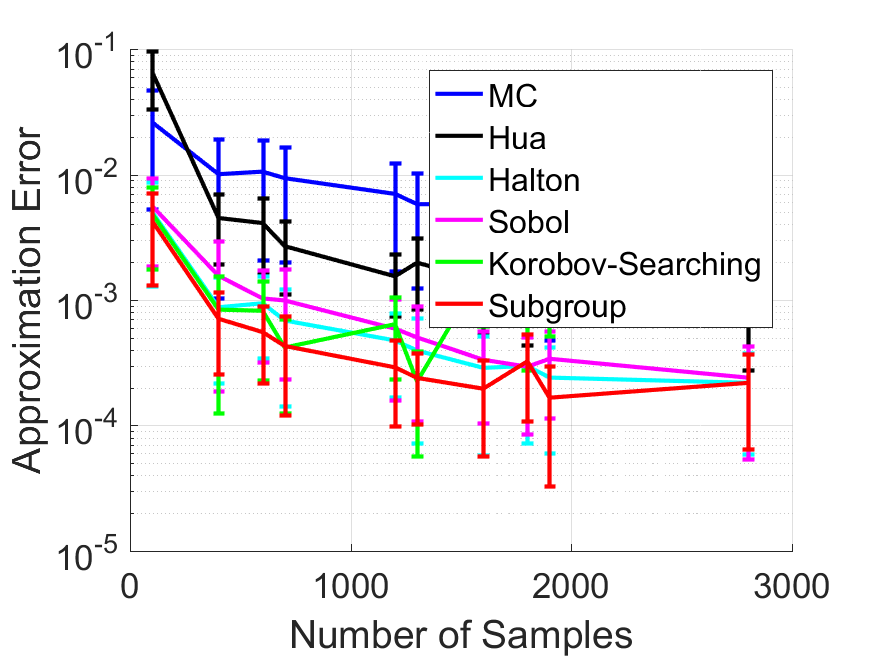}}
\subfigure[\scriptsize{ 100-d Integral Approximation }]{
\label{All50}
\includegraphics[width=0.224\linewidth]{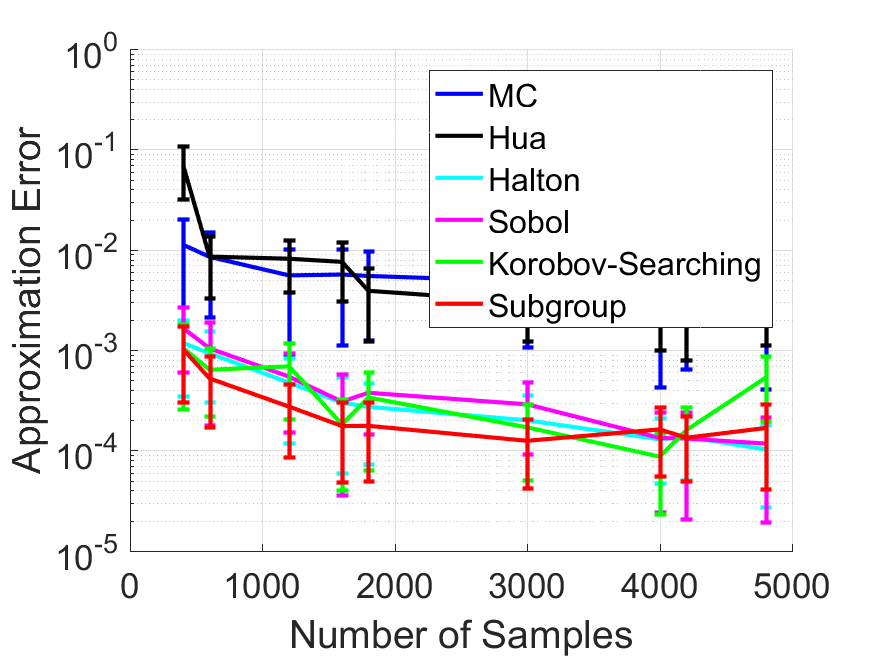}}
\subfigure[\scriptsize{ 500-d Integral Approximation }]{
\label{All500}
\includegraphics[width=0.224\linewidth]{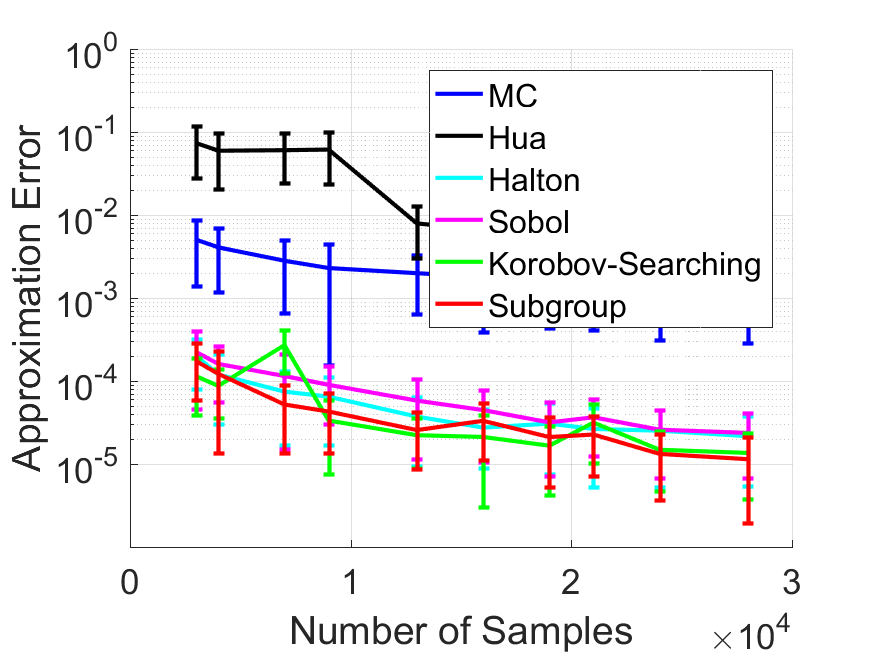}}
\subfigure[\scriptsize{ 1000-d Integral Approximation }]{
\label{All100}
\includegraphics[width=0.224\linewidth]{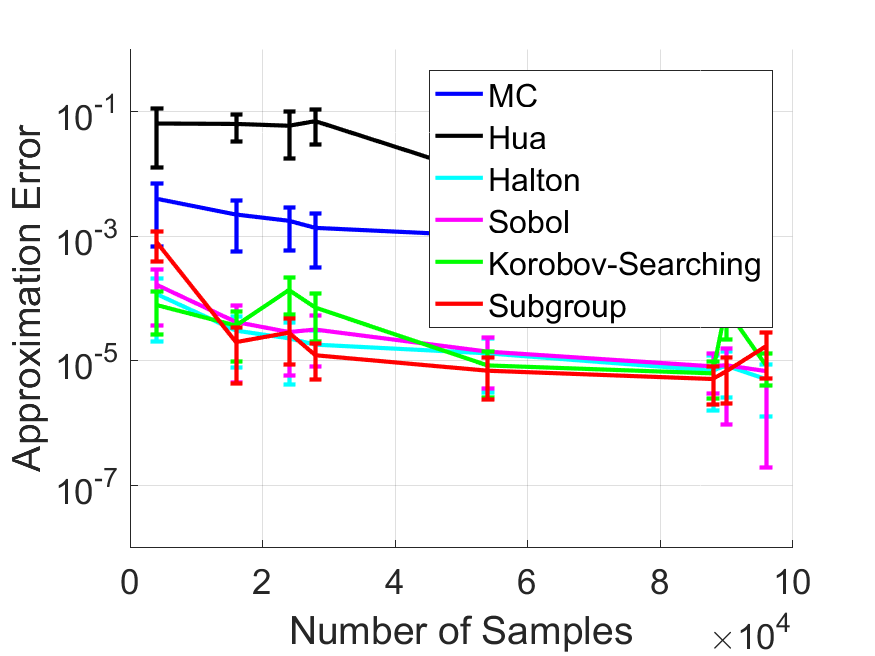}}
\caption{Mean approximation error  over 50 independent runs.error bars are with in $1\times$ std.}
\label{fig_App}
\vspace{-2mm}
\end{figure}

\subsection{Evaluation of the minimum distance}
\label{distance}
\vspace{-2mm}

We evaluate the minimum distance of our subgroup rank-1 lattice by comparing with Hua's method~\cite{hua2012applications} and the Korobov~\cite{korobov1960properties} searching method. We denote `SubGroup' as our subgroup rank-1 lattice, `Hua' as rank-1 lattice constructed by Hua's method~\cite{hua2012applications}, and `Korobov' as rank-1 lattice constructed by exhaustive computer search in Korobov form~\cite{korobov1960properties}.  

We set the dimension $d$ as in $\{50,100,200,500\}$. For each dimension $d$, we set the number of points $n$ as the first ten prime numbers such that $2d$ divides $n\!-\!1$, i.e., $2d \big | (n\!-\!1)$. The minimum $l_1$-norm-based toroidal distance and the minimum $l_2$-norm-based toroidal distance for each dimension are reported in Table~\ref{tab:l1normD} and  Table~\ref{tab:l2normD},  respectively.  The  larger the distance, the better. 

We can observe that our subgroup rank-1 lattice achieves consistently better (larger) minimum distances than Hua's method in all the cases. Moreover, we see that subgroup rank-1 lattice obtains, in 20 out of 40 cases, the same  $l_2$-norm-based toroidal distance and in 24 out of 40 cases the same  $l_1$-norm-based toroidal distance compared with the exhaustive computer search in Korobov form. The experiments show that our subgroup rank-1 lattice achieves the optimal toroidal distance in exhaustive computer searches in Korobov form in over half of all the cases.  Furthermore, the experimental result shows that our subgroup rank-1 lattice obtains a competitive distance compared with the exhaustive Korobov search in the remaining cases. Note that our subgroup rank-1 lattice is a closed-form construction which does not require computer search, making our method more appealing and simple to use. 

\textbf{ Time Comparison of Korobov searching and our sub-group rank-1 lattice. }
The table below shows the time cost (seconds) for lattice construction. The run time for Korobov searching grows fast to hours. Our method can run in less than one second, achieving a $10^4 \times$ to $10^5 \times$ speed-up. The speed-up increases when $n$ and $d$ becomes larger.  

\begin{table}[h]
\centering
\begin{footnotesize}
\label{tab:l1normD}
\resizebox{\columnwidth}{!}{
\begin{tabular}{ccccccccccccc}
\toprule

\multirow{4}{*}{\begin{tabular}[c]{@{}l@{}} d=500 \end{tabular}}
      &     & n=3001            & 4001            & 7001            & 9001            & 13001       & 16001           & 19001           & 21001           & 24001           & 28001         \\  
& SubGroup & 0.0185  &   0.0140  &   0.0289  &    0.043  &   0.0386      & 0.0320   &  0.0431  &   0.0548  &   0.0562  &   0.0593  \\ 
& Korobov  &  34.668   &     98.876    &   152.86      &  310.13   &    624.56  &      933.54    &   1308.9   &    1588.5      & 2058.5   &    2815.9 \\ 

\midrule

\multirow{4}{*}{\begin{tabular}[c]{@{}l@{}} d=1000 \end{tabular}}
      &     & n=4001  &     16001   &    24001   &    28001   &    54001        & 70001    &   76001   &    88001  &     90001  &     96001        \\  
& SubGroup & 0.0388  &    0.0618  &    0.1041   &   0.1289   &    0.2158  &  0.2923   &   0.3521   &   0.4099   &    0.5352   &  0.5663  \\ 
& Korobov  &  112.18    &   1849.4  &      4115.9   &    5754.6   &     20257   &   34842    &    43457     &   56798   &     56644    &    69323 \\ 

 \bottomrule
\end{tabular}
}
\end{footnotesize}
\vspace{-3mm}
\end{table}
\vspace{-3mm}

\subsection{Integral approximation}
\label{testInapp}

We evaluate our subgroup rank-1 lattice on the integration test problem 
\begin{align}
 &   f(\boldsymbol{x}):= \exp{\Big(c\sum_{j=1}^d x_jj^{-b} \Big) } \\
 &   I(f):= \int _{[0,1]^d}  {f(\boldsymbol{x})d\boldsymbol{x}} = \displaystyle\prod_{j=1}^{d} \frac{\exp(cj^{-b})-1}{cj^{-b}}.
\end{align}
We compare with i.i.d. Monte Carlo, a Hua's rank-1 lattice~\cite{hua2012applications}, Korobov searching rank-1 lattice~\cite{korobov1959approximate},  Halton sequence, and Sobol sequence ~\cite{dick2013high}.  For both  Halton sequence and Sobol sequence, we use the scrambling technique suggested in~\cite{dick2013high}. For all the QMC methods,  we use the random shift technique as in Eq.(\ref{randomshift}).


We fix $b=2$ and $c=1$ in all the experiments. We set dimension $d=100$ and $d=500$, respectively. We set the number of points $n$ as the first ten prime numbers such that $2d$ divides $n\!-\!1$, i.e., $2d \big | (n\!-\!1)$.

The mean approximation error ($\frac{|Q(f)-I(f)|}{|I(f)|}$) with error bars over 50 independent runs for each dimension $d$ is presented in Figure~\ref{fig_App}. 
 We can see that Hua's method obtains a smaller error than i.i.d Monte Carlo on the 50-d problem, however, it becomes worse than MC on 500-d and 1000-d problems. Moreover,  our subgroup rank-1 lattice  obtains a consistent smaller error on all the tested problems than Hua and MC.  In addition,  our subgroup rank-1 lattice achieves a slightly better performance than Halton, Sobol and  Korobov searching method.


\begin{figure}[t]
\centering
\subfigure[\scriptsize{$\frac{{{{\left\| {\widetilde K - K} \right\|}_F}}}{{{{\left\| K \right\|}_F}}}$ for Gaussian Kernel }]{
\label{fig2a_K_meanG}
\includegraphics[width=0.3\linewidth]{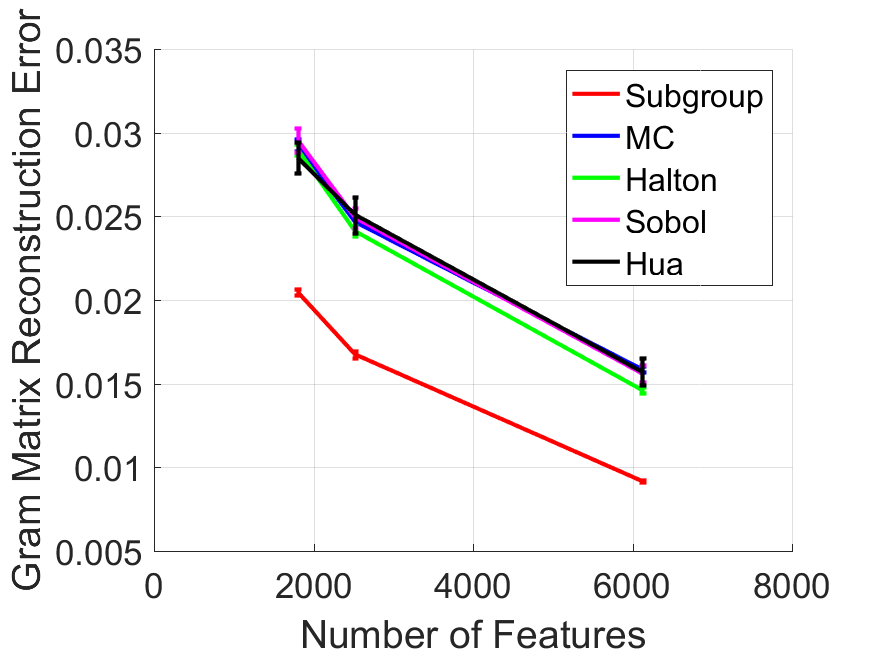}}
\subfigure[\scriptsize{$\frac{{{{\left\| {\widetilde K - K} \right\|}_F}}}{{{{\left\| K \right\|}_F}}}$ for First-order Arc Kernel}]{
\label{fig2c_K_acrCosine}
\includegraphics[width=0.3\linewidth]{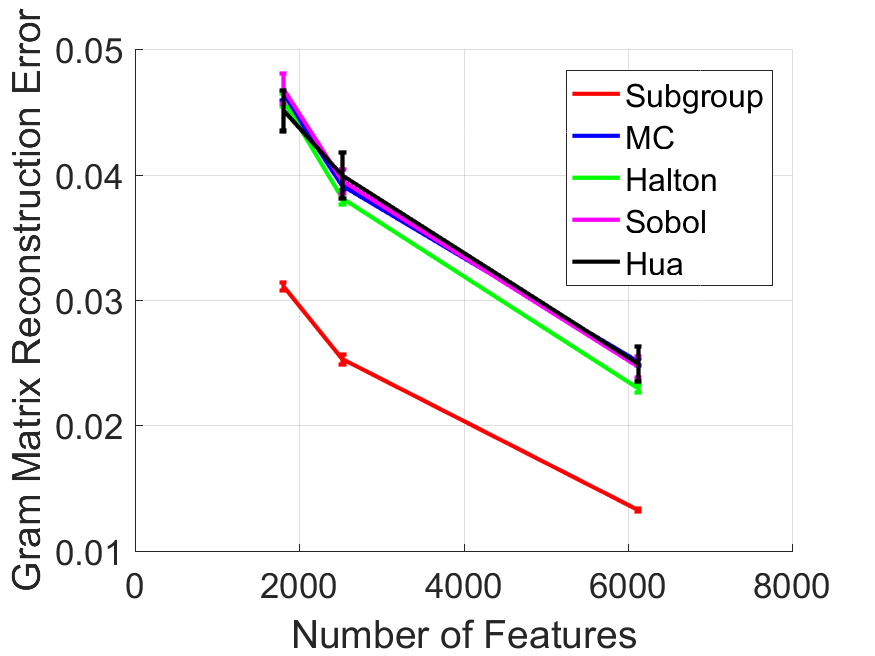}}
\subfigure[\scriptsize{$\frac{{{{\left\| {\widetilde K - K} \right\|}_F}}}{{{{\left\| K \right\|}_F}}}$ for  Zero-order Arc Kernel}]{
\label{fig2c_K_meanAngle}
\includegraphics[width=0.3\linewidth]{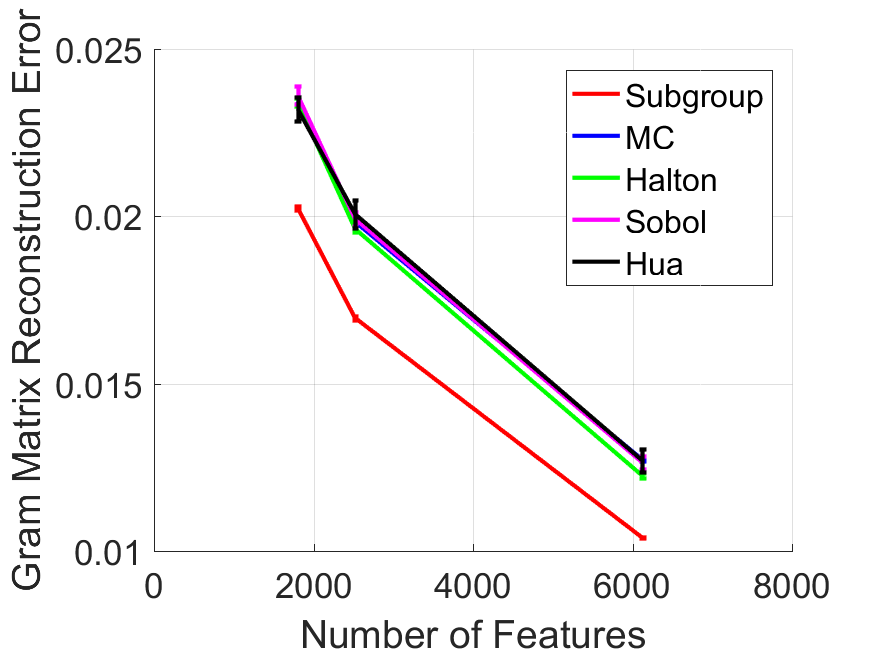}}
\subfigure[\scriptsize{$\frac{{{{\left\| {\widetilde K - K} \right\|}_\infty }}}{{{{\left\| K \right\|}_\infty }}}$ for Gaussian Kernel}]{
\label{fig2c_K_Max_Gaussian}
\includegraphics[width=0.3\linewidth]{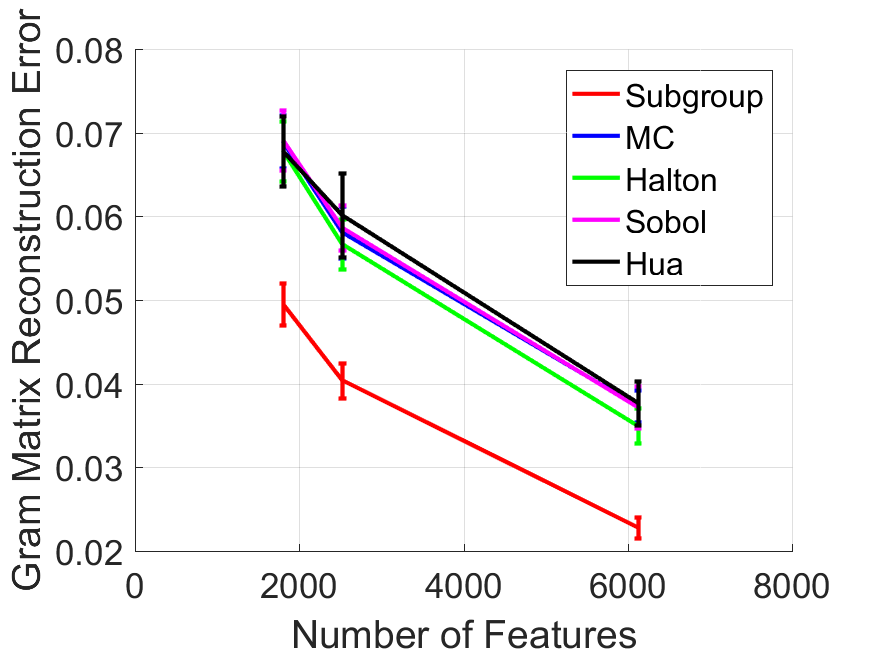}}
\subfigure[\scriptsize{$\frac{{{{\left\| {\widetilde K - K} \right\|}_\infty }}}{{{{\left\| K \right\|}_\infty }}}$  for First-order Arc Kerne}]{
\label{fig2c_K_Max_acrCosine}
\includegraphics[width=0.3\linewidth]{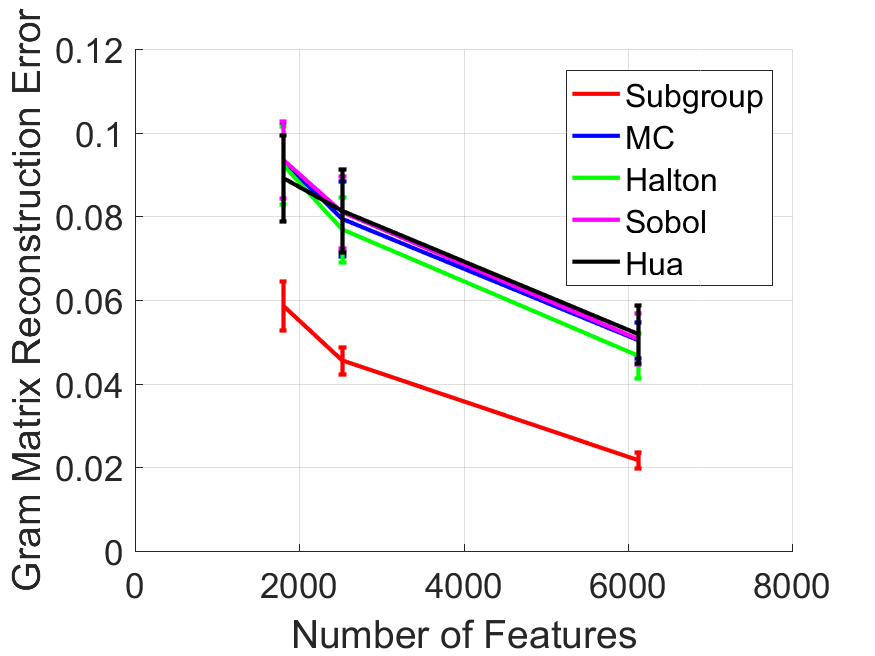}}
\subfigure[\scriptsize{$\frac{{{{\left\| {\widetilde K - K} \right\|}_\infty }}}{{{{\left\| K \right\|}_\infty }}}$  for Zero-order Arc Kernel}]{
\label{fig2c_K_maxAngle}
\includegraphics[width=0.3\linewidth]{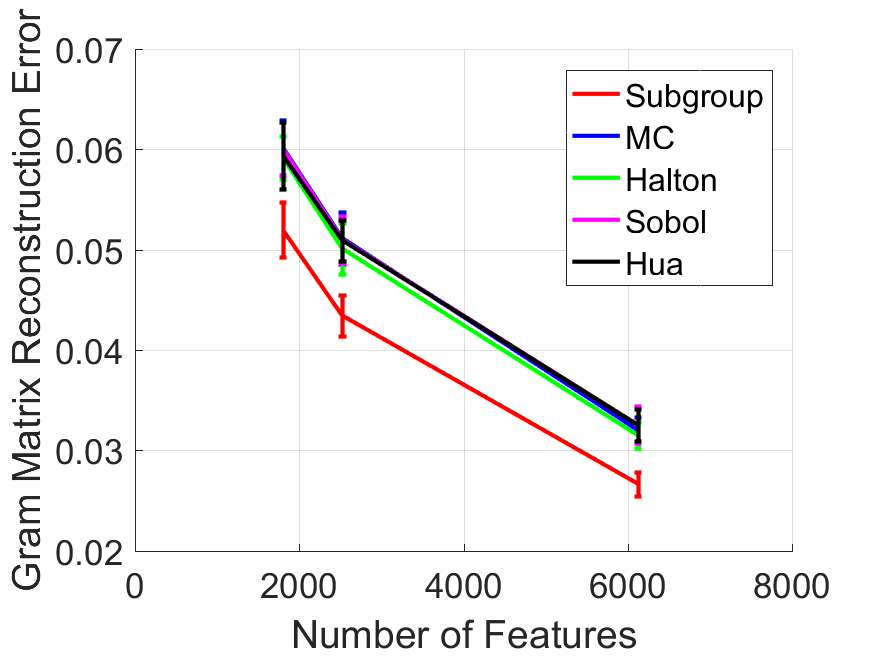}}
\caption{Relative Mean and Max Reconstruction Error for Gaussian, Zero-order and First-order Arc-cosine Kernel on DNA dataset. Error bars are within $1 \times$ std.}
\label{fig_KernelApp}
\end{figure}

\subsection{Kernel approximation}
\label{kerApp}

We evaluate the performance of subgroup rank-1 lattice on kernel approximation tasks by comparing with other QMC baseline methods. We test the kernel approximation of the Gaussian kernel, the zeroth-order arc-cosine kernel, and the first-order arc-cosine kernel as in~\cite{choromanski2016recycling}. 

We compare subgroup rank-1 lattice with a Hua's rank-1 lattice~\cite{hua2012applications}, Halton sequence,  Sobol sequence ~\cite{dick2013high} and standard i.i.d.  Monte Carlo sampling.  For both the Halton sequence and Sobol sequence, we use the scrambling technique suggested in~\cite{dick2013high}. For both  subgroup rank-1 lattice and Hua's rank-1 lattice, we use the random shift as in Eq.(\ref{randomshift}).
We evaluate the methods on the DNA~\cite{pourkamali2018randomized} and the SIFT1M~\cite{jegou2010product} dataset over 50 independent runs. Each run contains 2000 random  samples to construct the Gram matrix. The bandwidth  parameter of Gaussian kernel is set to 15 in all the experiments.

The mean Frobenius norm  approximation error (${\| {\widetilde K \!-\! K} \|}_F /{\left\| K \right\|}_F  $) and maximum norm approximation error  (${\| {\widetilde K \!-\! K} \|}_\infty / {\| K \|}_\infty$) with error bars on DNA~\cite{pourkamali2018randomized}  dataset  are plotted in Figure~\ref{fig_KernelApp}. The results on  SIFT1M~\cite{jegou2010product} is given  in Figure~6 in the supplement. 
The experimental result shows that subgroup rank-1 lattice consistently  obtains a smaller approximation error compared with other baselines. 
\subsection{Approximation on Graphical Model}
\vspace{-2mm}

For general Boltzmann machines with continuous state in $[0,1]$, the energy function of $\boldsymbol{x} \in [0,1]^d $ is defined as $E(\boldsymbol{x})=  -(\boldsymbol{x}^\top\boldsymbol{W}\boldsymbol{x} + \boldsymbol{b}^\top \boldsymbol{x})/d $. The normalization constant is $Z= \int _{[0,1]^d} { \exp{(-E(\boldsymbol{x})} )  d\boldsymbol{x}  }  $.  For inference, the marginal likelihood of  observation $\boldsymbol{v} \in \mathbb{R}^d$ is 
$\mathcal{L}(\boldsymbol{v})= \int _{[0,1]^d} { \exp{(-f(\boldsymbol{v}))} \exp{(-E(\boldsymbol{h})} )/Z \text{d}\boldsymbol{h}  } $ with function $f(\boldsymbol{v}) = -(\boldsymbol{v}^\top\boldsymbol{W}_v\boldsymbol{v} + 2 \boldsymbol{v}^\top\boldsymbol{W}_{h} \boldsymbol{h}  +\boldsymbol{b}_v^\top \boldsymbol{v})/d $, where $\boldsymbol{h} \in \mathbb{R}^d$ denotes the hidden states.

We  evaluate our method on approximation of the normalization constant and inference by comparing with i.i.d. Monte Carlo (MC), slice sampling (SS) and Hamiltonian Monte Carlo (HMC).
We generate the elements of $\boldsymbol{W}$, $\boldsymbol{W}_v$, $\boldsymbol{W}_h$, $\boldsymbol{b}$ and $\boldsymbol{b}_v$ by sampling from standard Gaussian $\mathcal{N}(0,1)$. These parameters are fixed and kept the same for  all the methods in comparison. For inference, we generate an observation  $\boldsymbol{v} \in [0,1]^d $ by uniformly sampling  and keep it fixed and same for all the methods. For SS and HMC, we use the \textit{slicesample} function and \textit{hmcSampler} function in MATLAB, respectively. We use the approximation of i.i.d. MC with $10^7$ samples as the pseudo ground-truth. The approximation errors $|\widehat{Z}- Z |/Z$ and  $|\widehat{\mathcal{L}}- \mathcal{L} |/\mathcal{L} $ are shown in Fig.\ref{z10}-\ref{z500} and Fig.\ref{Infer10}-\ref{Infer500}, respectively. our method consistently outperforms MC, HMC and SS on all cases.  Moreover, our method is much cheaper than SS and HMC.

\begin{figure*}[t]
\centering
\subfigure[\scriptsize{10-d$|\widehat{Z}- Z |/Z$}]{
\label{z10}
\includegraphics[width=0.224\linewidth]{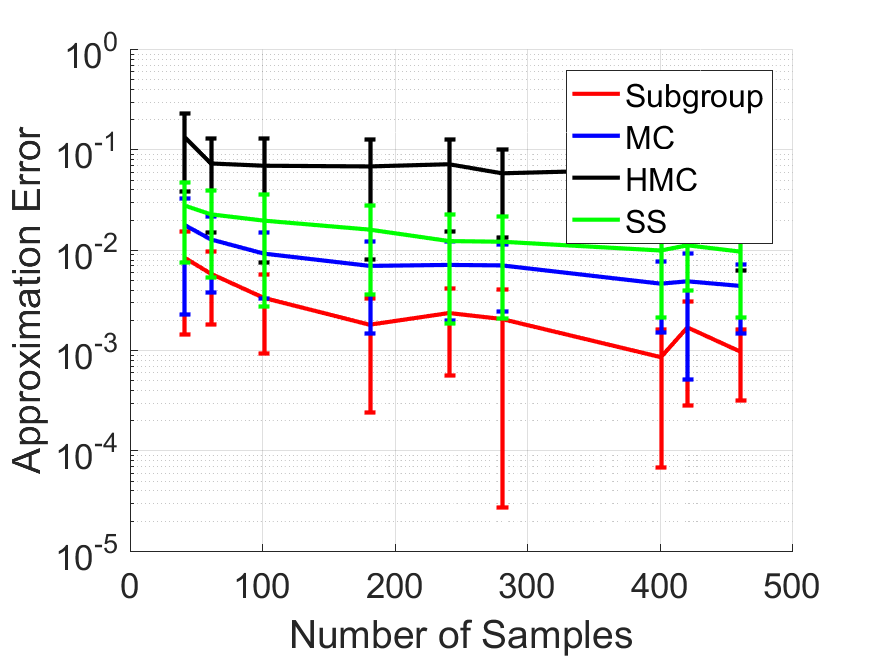}}
\subfigure[\scriptsize{50-d $|\widehat{Z}- Z |/Z$}]{
\label{z50}
\includegraphics[width=0.224\linewidth]{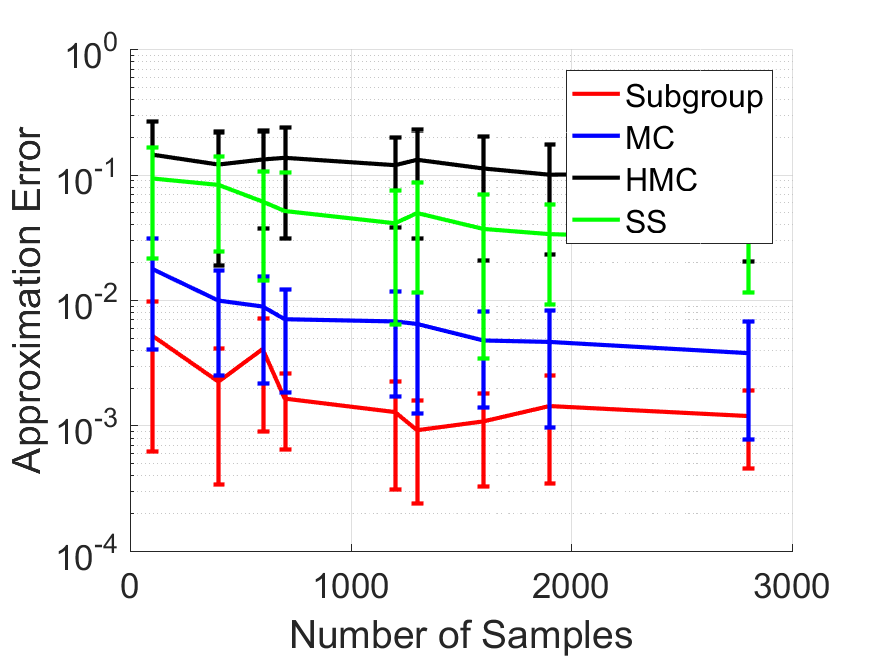}}
\subfigure[\scriptsize{100-d $|\widehat{Z}- Z |/Z$}]{
\label{z100}
\includegraphics[width=0.224\linewidth]{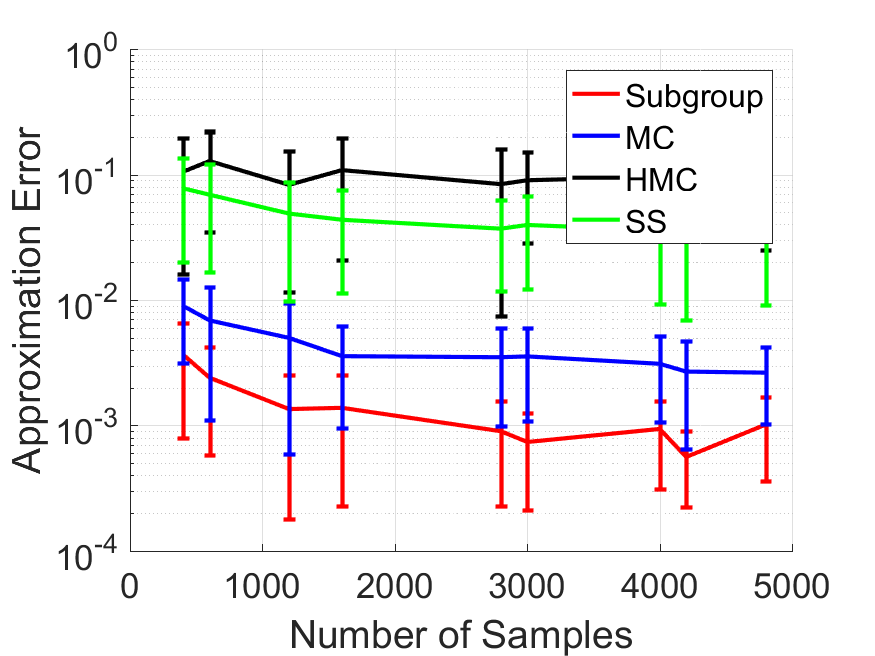}}
\subfigure[\scriptsize{500-d $|\widehat{Z}- Z |/Z$}]{
\label{z500}
\includegraphics[width=0.224\linewidth]{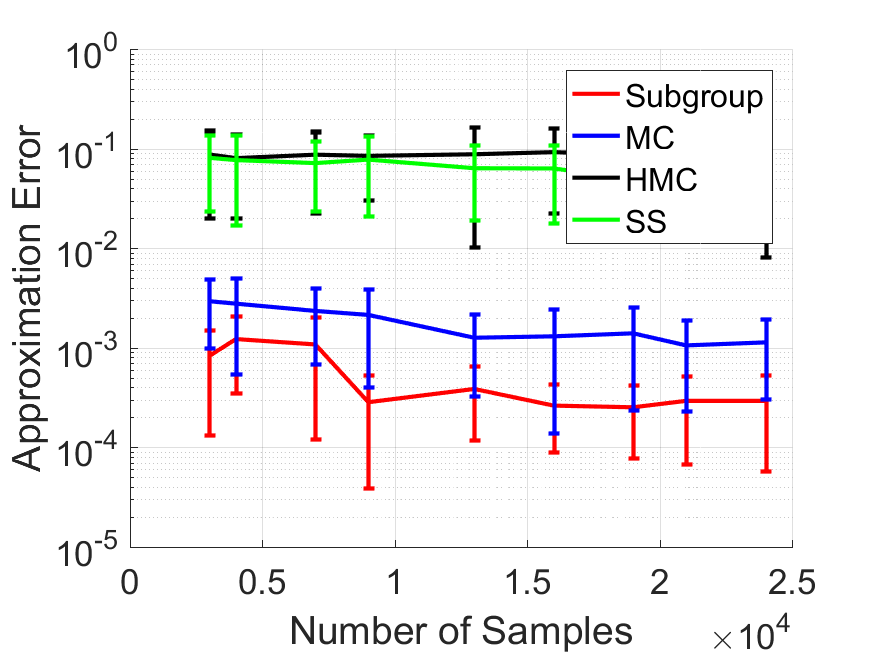}}
\subfigure[\scriptsize{10-d  $|\widehat{\mathcal{L}}- \mathcal{L} |/\mathcal{L} $}]{
\label{Infer10}
\includegraphics[width=0.224\linewidth]{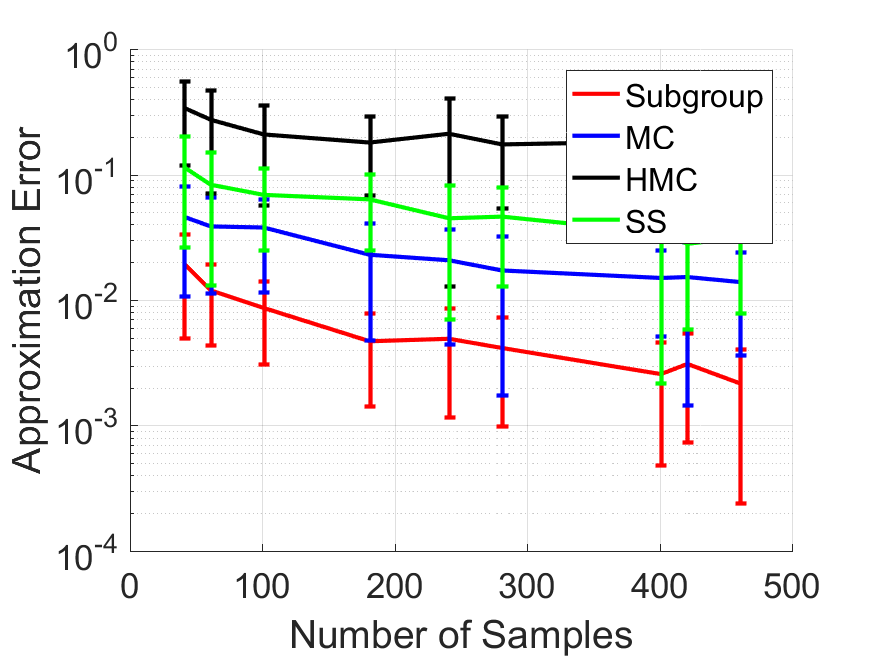}}
\subfigure[\scriptsize{50-d $|\widehat{\mathcal{L}}- \mathcal{L} |/\mathcal{L} $}]{
\label{Infer50}
\includegraphics[width=0.224\linewidth]{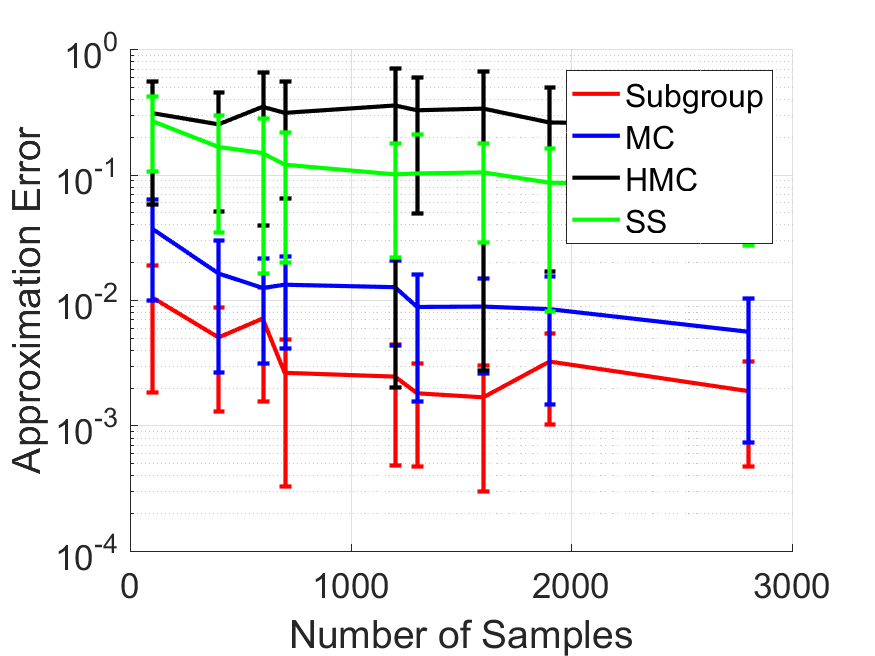}}
\subfigure[\scriptsize{100-d $|\widehat{\mathcal{L}}- \mathcal{L} |/\mathcal{L} $}]{
\label{Infer100}
\includegraphics[width=0.224\linewidth]{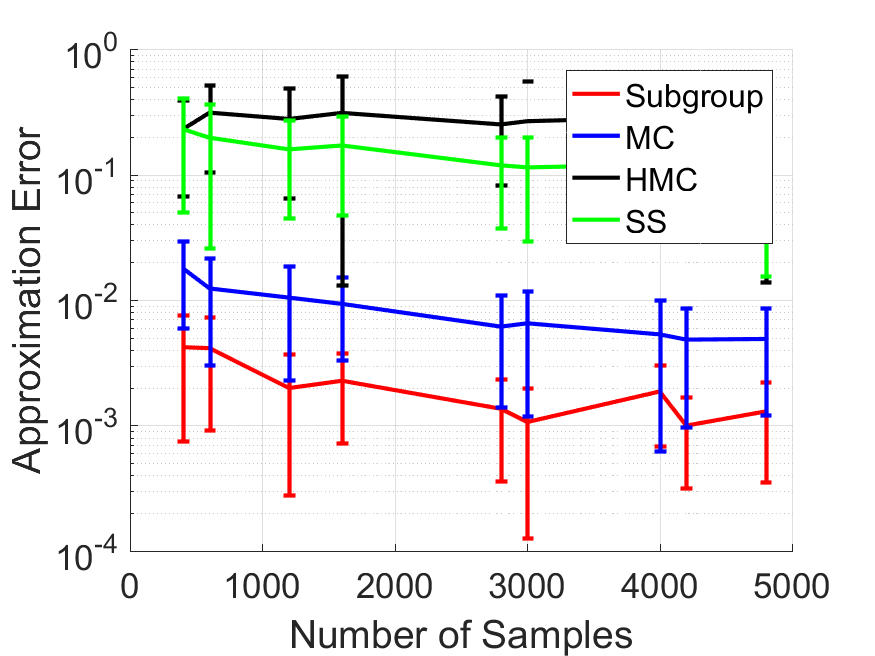}}
\subfigure[\scriptsize{500-d $|\widehat{\mathcal{L}}- \mathcal{L} |/\mathcal{L} $}]{
\label{Infer500}
\includegraphics[width=0.224\linewidth]{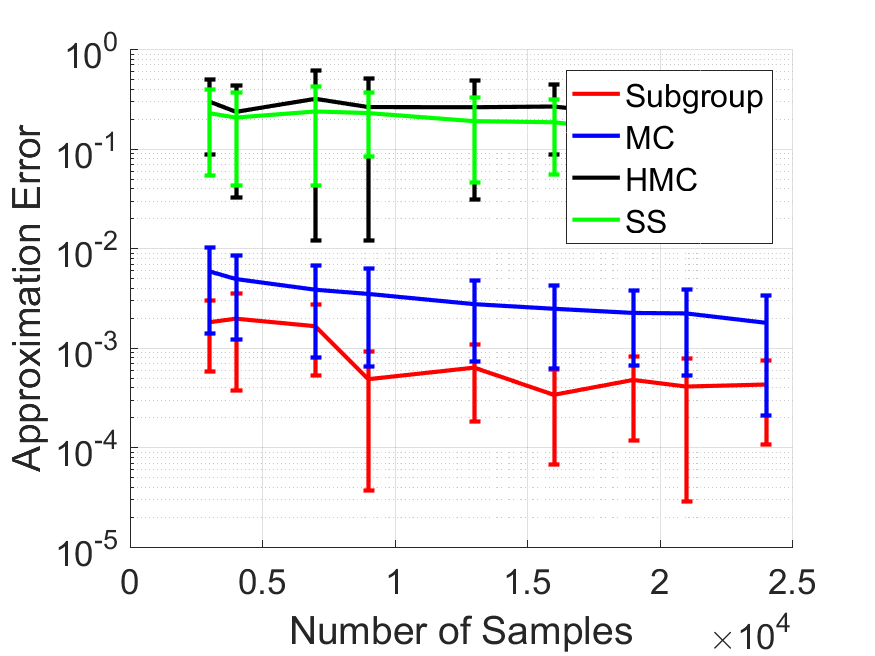}}

\caption{  Mean approximation error  over 50 independent runs. Error bars are with in $1\times$ std  }
\label{BMtest}
\vspace{-4mm}
\end{figure*}

\textbf{Comparison to sequential Monte Carlo.}
When the positive density region takes a large fraction of the entire domain, our method is very competitive. When it is only inside a small part of a large domain, our method may not be better than sequential adaptive sampling. In this case, it is interesting to take advantage of both lattice and adaptive sampling. E.g., one can employ our subgroup rank-1 lattice as a rough partition of the domain to find high mass regions, then take sequential adaptive sampling on the promising regions with the lattice points as the start points.  Also, it is interesting to consider recursively apply our subgroup rank-1 lattice to refine the partition. Moreover, our subgroup-based rank-1 lattice enables black-box evaluation without the need for gradient information.  In contrast, several sequential sampling methods, e.g., HMC,  need a gradient of density function for sampling.

\vspace{-3mm}
\section{Conclusion}
\vspace{-3mm}

We propose a closed-form method for rank-1 lattice construction, which is simple and efficient without exhaustive  computer search.  Theoretically, we prove that our subgroup rank-1 lattice has few different pairwise distance values, which is more regular to be evenly spaced. Moreover, we prove a lower and an upper bound for the minimum  toroidal distance of a non-degenerate rank-1 lattice. Empirically, our
subgroup rank-1 lattice obtains near-optimal minimum toroidal distance compared with Korobov exhaustive search.     Moreover, subgroup rank-1 lattice achieves  smaller integration approximation error. In addition, we propose a closed-form method to generate QMC points set on sphere $\mathbb{S}^{d-1}$. We proved upper bounds of the mutual coherence of the  generated points.  Further,  we show an example of  CycleGAN training in the supplement. Our  subgroup rank-1 lattice sampling and QMC on sphere can serve as an alternative for training generative models. 

\vspace{-3mm}
\section*{Acknowledgement and Funding Disclosure}
\vspace{-3mm}
We thank the reviewers for their valuable comments and suggestions.  Yueming Lyu was supported by UTS President Scholarship.
Prof. Ivor W. Tsang was supported by ARC DP180100106 and
DP200101328.

\section*{Broader Impact}

In this paper, we proposed a closed-form rank-one lattice construction based on group theory for Quasi-Monte Carlo.  Our method does not require the time-consuming exhaustive computer search.   Our method is a fundamental tool for integral approximation and sampling. 

 Our method may serve as a potential advance in QMC, which may have an impact on a wide range of applications that rely on integral approximation. It includes kernel approximation with feature map, variational inference in Bayesian learning,    generative modeling, and variational autoencoders.
 This may bring useful applications and be beneficial to society and the community.  Since our method focuses more on the theoretical side, the direct negative influences and ethical issues are negligible.

\bibliography{paper}

\begin{thebibliography}{10}

\bibitem{arouna2004adaptative}
Bouhari Arouna.
\newblock Adaptative monte carlo method, a variance reduction technique.
\newblock {\em Monte Carlo Methods and Applications}, 10(1):1--24, 2004.

\bibitem{avron2016quasi}
Haim Avron, Vikas Sindhwani, Jiyan Yang, and Michael~W Mahoney.
\newblock Quasi-monte carlo feature maps for shift-invariant kernels.
\newblock {\em The Journal of Machine Learning Research}, 17(1):4096--4133,
  2016.

\bibitem{beal2003variational}
Matthew~James Beal et~al.
\newblock {\em Variational algorithms for approximate Bayesian inference}.
\newblock university of London London, 2003.

\bibitem{bourgain2006estimates}
Jean Bourgain, Alexey~A Glibichuk, and SERGEI~VLADIMIROVICH KONYAGIN.
\newblock Estimates for the number of sums and products and for exponential
  sums in fields of prime order.
\newblock {\em Journal of the London Mathematical Society}, 73(2):380--398,
  2006.

\bibitem{buchholz2018quasi}
Alexander Buchholz, Florian Wenzel, and Stephan Mandt.
\newblock Quasi-monte carlo variational inference.
\newblock {\em arXiv preprint arXiv:1807.01604}, 2018.

\bibitem{choromanski2016recycling}
Krzysztof Choromanski and Vikas Sindhwani.
\newblock Recycling randomness with structure for sublinear time kernel
  expansions.
\newblock 2016.

\bibitem{pack}
Sabrina Dammertz and Alexander Keller.
\newblock Image synthesis by rank-1 lattices.
\newblock In {\em Monte Carlo and Quasi-Monte Carlo Methods 2006}, 2008.

\bibitem{dick2013high}
Josef Dick, Frances~Y Kuo, and Ian~H Sloan.
\newblock High-dimensional integration: the quasi-monte carlo way.
\newblock {\em Acta Numerica}, 22:133--288, 2013.

\bibitem{doerr2013constructing}
Carola Doerr and Fran{\c{c}}ois-Michel De~Rainville.
\newblock Constructing low star discrepancy point sets with genetic algorithms.
\newblock In {\em Proceedings of the 15th annual conference on Genetic and
  evolutionary computation}, pages 789--796, 2013.

\bibitem{dummit2004abstract}
David~Steven Dummit and Richard~M Foote.
\newblock {\em Abstract algebra}, volume~3.
\newblock Wiley Hoboken, 2004.

\bibitem{minTdistance}
Leonhard Gr{\"u}nschlo{\ss}, Johannes Hanika, Ronnie Schwede, and Alexander
  Keller.
\newblock (t, m, s)-nets and maximized minimum distance.
\newblock In {\em Monte Carlo and Quasi-Monte Carlo Methods 2006}, pages
  397--412. Springer, 2008.

\bibitem{hammersley2013monte}
John Hammersley.
\newblock {\em Monte carlo methods}.
\newblock Springer Science \& Business Media, 2013.

\bibitem{hua2012applications}
L-K Hua and Yuan Wang.
\newblock {\em Applications of number theory to numerical analysis}.
\newblock Springer Science \& Business Media, 2012.

\bibitem{jegou2010product}
Herve Jegou, Matthijs Douze, and Cordelia Schmid.
\newblock Product quantization for nearest neighbor search.
\newblock {\em IEEE transactions on pattern analysis and machine intelligence},
  33(1):117--128, 2010.

\bibitem{kingma2013auto}
Diederik~P Kingma and Max Welling.
\newblock Auto-encoding variational bayes.
\newblock {\em arXiv preprint arXiv:1312.6114}, 2013.

\bibitem{korobov1959approximate}
AN~Korobov.
\newblock The approximate computation of multiple integrals.
\newblock In {\em Dokl. Akad. Nauk SSSR}, volume 124, pages 1207--1210, 1959.

\bibitem{korobov1960properties}
Nikolai~Mikhailovich Korobov.
\newblock Properties and calculation of optimal coefficients.
\newblock In {\em Doklady Akademii Nauk}, volume 132, pages 1009--1012. Russian
  Academy of Sciences, 1960.

\bibitem{laimer2017combined}
Helene Laimer.
\newblock On combined component-by-component constructions of lattice point
  sets.
\newblock {\em Journal of Complexity}, 38:22--30, 2017.

\bibitem{le2013fastfood}
Quoc Le, Tam{\'a}s Sarl{\'o}s, and Alex Smola.
\newblock Fastfood-approximating kernel expansions in loglinear time.
\newblock In {\em Proceedings of the international conference on machine
  learning}, 2013.

\bibitem{l2016algorithm}
Pierre L'ecuyer and David Munger.
\newblock Algorithm 958: Lattice builder: A general software tool for
  constructing rank-1 lattice rules.
\newblock {\em ACM Transactions on Mathematical Software (TOMS)}, 42(2):1--30,
  2016.

\bibitem{leobacher2014introduction}
Gunther Leobacher and Friedrich Pillichshammer.
\newblock {\em Introduction to quasi-Monte Carlo integration and applications}.
\newblock Springer, 2014.

\bibitem{lyness2003notes}
James~N Lyness.
\newblock Notes on lattice rules.
\newblock {\em Journal of Complexity}, 19(3):321--331, 2003.

\bibitem{lyu2017spherical}
Yueming Lyu.
\newblock Spherical structured feature maps for kernel approximation.
\newblock In {\em Proceedings of the 34th International Conference on Machine
  Learning-Volume 70}, pages 2256--2264, 2017.

\bibitem{makhzani2015adversarial}
Alireza Makhzani, Jonathon Shlens, Navdeep Jaitly, Ian Goodfellow, and Brendan
  Frey.
\newblock Adversarial autoencoders.
\newblock {\em arXiv preprint arXiv:1511.05644}, 2015.

\bibitem{niederreiter1992random}
Harald Niederreiter.
\newblock {\em Random number generation and quasi-Monte Carlo methods},
  volume~63.
\newblock Siam, 1992.

\bibitem{nuyens2006fast}
Dirk Nuyens and Ronald Cools.
\newblock Fast algorithms for component-by-component construction of rank-1
  lattice rules in shift-invariant reproducing kernel hilbert spaces.
\newblock {\em Mathematics of Computation}, 75(254):903--920, 2006.

\bibitem{owen2019monte}
Art~B Owen.
\newblock Monte carlo book: the quasi-monte carlo parts.
\newblock 2019.

\bibitem{pourkamali2018randomized}
Farhad Pourkamali-Anaraki, Stephen Becker, and Michael~B Wakin.
\newblock Randomized clustered nystrom for large-scale kernel machines.
\newblock In {\em Thirty-Second AAAI Conference on Artificial Intelligence},
  2018.

\bibitem{rahimi2008random}
Ali Rahimi and Benjamin Recht.
\newblock Random features for large-scale kernel machines.
\newblock In {\em Advances in neural information processing systems}, 2007.

\bibitem{salimans2017evolution}
Tim Salimans, Jonathan Ho, Xi~Chen, Szymon Sidor, and Ilya Sutskever.
\newblock Evolution strategies as a scalable alternative to reinforcement
  learning.
\newblock {\em arXiv preprint arXiv:1703.03864}, 2017.

\bibitem{shkredov2014exponential}
Ilya~D Shkredov.
\newblock On exponential sums over multiplicative subgroups of medium size.
\newblock {\em Finite Fields and Their Applications}, 30:72--87, 2014.

\bibitem{sloan2002component}
I~Sloan and A~Reztsov.
\newblock Component-by-component construction of good lattice rules.
\newblock {\em Mathematics of Computation}, 71(237):263--273, 2002.

\bibitem{sloan2001tractability}
Ian~H Sloan and Henryk Wo{\'z}niakowski.
\newblock Tractability of multivariate integration for weighted korobov
  classes.
\newblock {\em Journal of Complexity}, 17(4):697--721, 2001.

\bibitem{tompkins2019black}
Anthony Tompkins, Ransalu Senanayake, Philippe Morere, and Fabio Ramos.
\newblock Black box quantiles for kernel learning.
\newblock In {\em The 22nd International Conference on Artificial Intelligence
  and Statistics}, pages 1427--1437, 2019.

\bibitem{yang2014quasi}
Jiyan Yang, Vikas Sindhwani, Haim Avron, and Michael Mahoney.
\newblock Quasi-monte carlo feature maps for shift-invariant kernels.
\newblock In {\em International Conference on Machine Learning}, pages
  485--493, 2014.

\bibitem{zhu2017unpaired}
Jun-Yan Zhu, Taesung Park, Phillip Isola, and Alexei~A Efros.
\newblock Unpaired image-to-image translation using cycle-consistent
  adversarial networks.
\newblock In {\em Proceedings of the IEEE international conference on computer
  vision}, pages 2223--2232, 2017.

\end{thebibliography}
\bibliographystyle{plain}

\newpage

\appendix


\textbf{Organization:}
In the supplement,  we present the detailed proof of  the Theorem 1, Theorem 2 and Corollary 1 in section~\ref{proofT1},  section~\ref{proofT2} and section~\ref{proofc1}, respectively. We then present a subgroup-based QMC on sphere $\mathbb{S}^{d-1}$ in Section~\ref{sQMC}. We give the detailed proof of Theorem~\ref{Sbound1} and Theorem~\ref{Sbound2} in section~\ref{proofT3} and section~\ref{proofT4}, respectively.  We then present QMC for generative CycleGAN in section~\ref{QMCgm} and section~\ref{GenerativeExp}. At last, we present the experimental results of kernel approximation on SIFT1M dataset in Figure~\ref{fig_KernelApp2}.


\section{Proof of Theorem 1}
\label{proofT1}

\begin{theorem*}
Suppose $n$ is a prime number and $2d | (n-1)$. Let $g$ be a  primitive root of $n$. Let $  \boldsymbol{z} = [ g^0,g^{\frac{n-1}{2d}}, g^{\frac{2(n-1)}{2d}},\cdots, g^{\frac{(d-1)(n-1)}{2d} } ] \; \text{mod} \; n $. Construct a rank-1 lattice $X \!=\!\{\boldsymbol{x}_0,\cdots,\boldsymbol{x}_{n-1}\}$ with $\boldsymbol{x}_i = \frac{i \boldsymbol{z}  \; \text{mod} \;n}{n}, i \in \{0,...,n-1\}  $. Then, there are $\frac{n-1}{2d}$ distinct pairwise  toroidal distance values among $X$, and for each distance value, there are the same number of pairs that obtain this value.
\end{theorem*}

\begin{proof}

From the definition of the rank-1 lattice, we know that 
\begin{align}
  \|\boldsymbol{x}_i - \boldsymbol{x}_j   \|_{T_p} =   \left \| \frac{i \boldsymbol{z}  \; \text{mod} \;n}{n} - \frac{j \boldsymbol{z}  \; \text{mod} \;n}{n} \right\|_{T_p} =   \left \| \frac{(i-j) \boldsymbol{z}  \; \text{mod} \;n}{n}   \right\|_{T_p} = \left \| \frac{k \boldsymbol{z}  \; \text{mod} \;n}{n}   \right\|_{T_p} =  \|\boldsymbol{x}_k  \|_{T_p}, 
\end{align}
where $ \|{\bf{x}} \|_{T_p} $ denotes the $l_p$-norm-based toroidal distance between $ \bf{x} $ and $ \bf{0} $, and $k \equiv i-j \; \text{mod} \;n $.

For non-identical pair $\boldsymbol{x}_i, \boldsymbol{x}_j \in  X \!=\!\{\boldsymbol{x}_0,\cdots,\boldsymbol{x}_{n-1}\}$,  we know $i \ne j$. Thus, $  i-j \equiv k \in \{1,\cdots, n-1\}$. Moreover, for each $k$, there are $n$ pairs of $i,j \in \{0,\cdots,n-1\}$ obtaining $i-j \equiv k \; \text{mod} \; n$. Therefore, the non-identical pairwise toroidal distance is determined by $\|\boldsymbol{x}_k  \|_{T_p} $ for $ k \in \{1,\cdots,n-1\}$. Moreover,  each $\|\boldsymbol{x}_k  \|_{T_p} $ corresponds to $n$ pairwise distances.  

From the definition of the $l_p$-norm-based toroidal distance, we know that
\begin{align}
     \|\boldsymbol{x}_k  \|_{T_p} & = \left\| \text{min} \left( \frac{k \boldsymbol{z}  \; \text{mod} \;n}{n},  \frac{ n- k \boldsymbol{z}  \; \text{mod} \;n}{n} \right)  \right\|_{p} \nonumber
   \\  & =  \left\| \text{min} \left( \frac{k \boldsymbol{z}  \; \text{mod} \;n}{n},  \frac{  (-k\boldsymbol{z})   \; \text{mod} \;n}{n} \right)  \right\|_{p}, 
 \end{align}
 where $\text{min}(\cdot,\cdot)$  denotes the element-wise min operation between two inputs.

Since  $n$ is a prime number, from the number theory, we know that for a primitive root $g$, the residue of $\{g^0,g^1,\cdots,g^{n-2} \} $  modulo $n$ forms a cyclic group under multiplication, and $g^{n-1} \equiv 1 \; \text{mod} \;n$.   Moreover, 
there is  a one-to-one correspondence between the residue of $\{g^0,g^1,\cdots,g^{n-2} \} $  modulo $n$ and the set $\{1,2,\cdots,n-1\}$. Then,  we know that $\exists k',  g^{k'} \equiv k \; \text{mod} \; n$. It follows that
 \begin{align}
 \label{xk}
      \|\boldsymbol{x}_k  \|_{T_p}  =  \left\| \text{min} \left( \frac{g^{k'} \boldsymbol{z}  \; \text{mod} \;n}{n},  \frac{  (-g^{k'}\boldsymbol{z})   \; \text{mod} \;n}{n} \right)  \right\|_{p}.
 \end{align}
 
 Since $ (g^{\frac{n-1}{2}})^2 =g^{n-1} \equiv 1 \; \text{mod} \;n$ and $g$ is a primitive root, we know  that $g^{\frac{n-1}{2}} \equiv -1 \; \text{mod} \;n$. Denote $ \{\boldsymbol{z},-\boldsymbol{z}\} :=\{{z}_1,z_2,\cdots,z_d,-z_1,z_2,\cdots,-z_d\} $.
  Since $  \boldsymbol{z} = [ g^0,g^{\frac{n-1}{2d}}, g^{\frac{2(n-1)}{2d}},\cdots, g^{\frac{(d-1)(n-1)}{2d} } ] \; \text{mod} \; n $, we know that
 \begin{align}
    \{\boldsymbol{z},-\boldsymbol{z}\} & \equiv  \{\boldsymbol{z},g^{\frac{n-1}{2}} \boldsymbol{z}\} \; \text{mod} \;n \\
     & \equiv  \{g^0,g^{\frac{n-1}{2d}}, g^{\frac{2(n-1)}{2d}},\cdots, g^{\frac{(d-1)(n-1)}{2d} }, g^{\frac{n-1}{2}+0},g^{\frac{n-1}{2}+\frac{n-1}{2d}},\cdots,g^{\frac{n-1}{2}+\frac{(d-1)(n-1)}{2d}}  \}  \; \text{mod} \;n  \\
     & \equiv \{g^0,g^{\frac{n-1}{2d}}, g^{\frac{2(n-1)}{2d}},\cdots, g^{\frac{(d-1)(n-1)}{2d} }, g^{\frac{d(n-1)}{2d}},g^{\frac{(d+1)(n-1)}{2d}},\cdots,g^{\frac{(2d-1)(n-1)}{2d}}  \}  \; \text{mod} \;n.
 \end{align}
 It follows that $H :=\{{z}_1,z_2,\cdots,z_d,-z_1,z_2,\cdots,-z_d\} \; \text{mod} \;n$ forms a subgroup of the group $\{g^0,g^1,\cdots,g^{n-2} \} \; \text{mod} \;n$.  From Lagrange's theorem in group theory~\cite{dummit2004abstract},  we know that the cosets of the subgroup $H$ partition the entire group $\{g^0,g^1,\cdots,g^{n-2}\}$ into equal-size, non-overlapping sets, i.e., cosets $g^0H, g^1H, \cdots, g^{\frac{n-1-2d}{2d}}H$,  and the number of cosets of $H$ is $\frac{n-1}{2d}$.
 
 Together with Eq.(\ref{xk}), we know that distance $ \|\boldsymbol{x}_k  \|_{T_p}$  for $k' \in \{0,\cdots,n-2\} $ has  $\frac{n-1}{2d}$ different values simultaneously hold for all $p \in (0,\infty)$, i.e, $ \left\| \text{min} \left( \frac{g^{h} \boldsymbol{z}  \; \text{mod} \;n}{n},  \frac{  (-g^{h}\boldsymbol{z})   \; \text{mod} \;n}{n} \right)  \right\|_{p}$ for  $h \in \{0,\cdots,\frac{n-1}{2d}-1\}$ . And for each distance value, there are the same number of terms $ \|\boldsymbol{x}_k  \|_{T_p}$ that obtain this value. 
 Since each $ \|\boldsymbol{x}_k  \|_{T_p}$ corresponds to $n$ pairwise distance $ \|\boldsymbol{x}_i - \boldsymbol{x}_j   \|_{T_p}$, where $k \equiv i-j \; \text{mod} \; n$, there are $\frac{n-1}{2d}$ distinct  pairwise toroidal distance. Moreover, for each distance value, there are the same number of pairs that obtain this value.

\end{proof}

\section{Proof of Theorem 2}
\label{proofT2}

\begin{theorem*}
Suppose $n$ is a prime number and $n \ge 2d+1$.  Let $  \boldsymbol{z} = [ z_1,z_2, \cdots, z_d  ] $ with $1 \le z_k \le n-1 $.  Construct non-degenerate rank-1 lattice $X=\{\boldsymbol{x}_0,\cdots,\boldsymbol{x}_{n-1}\}$ with $\boldsymbol{x}_i = \frac{i \boldsymbol{z}  \; \text{mod} \;n}{n}, i \in \{0,...,n-1\}  $. Then, the minimum pairwise toroidal distance can be bounded as 
\begin{align}
&  \frac{d(d+1)}{2n}  \le \min_{  i,j \in \{0,\cdots,n-1\} ,  i \ne j} \|{\bf{x}}_i - {\bf{x}}_j \|_{T_1} \le \frac{(n+1)d}{4n}  \\
 &   \frac{\sqrt{6d(d+1)(2d+1)}}{6n}  \le \min_{  i,j \in \{0,\cdots,n-1\} ,  i \ne j} \|{\bf{x}}_i - {\bf{x}}_j \|_{T_2}  \le  \sqrt{\frac{(n+1)d}{12n} }, 
\end{align}
where $\|\cdot\|_{T_1}$ and  $\|\cdot\|_{T_2}$ denotes the $l_1$-norm-based toroidal distance and the $l_2$-norm-based toroidal distance, respectively.
\end{theorem*}

\begin{proof}

From the definition of the rank-1 lattice, we know that 
\begin{align}
  \|\boldsymbol{x}_i - \boldsymbol{x}_j   \|_{T_p} =   \left \| \frac{i \boldsymbol{z}  \; \text{mod} \;n}{n} - \frac{j \boldsymbol{z}  \; \text{mod} \;n}{n} \right\|_{T_p} =   \left \| \frac{(i-j) \boldsymbol{z}  \; \text{mod} \;n}{n}   \right\|_{T_p} = \left \| \frac{k \boldsymbol{z}  \; \text{mod} \;n}{n}   \right\|_{T_p} =  \|\boldsymbol{x}_k  \|_{T_p}, 
\end{align}
where $ \|{\bf{x}} \|_{T_p} $ denotes the $l_p$-norm-based toroidal distance, we know that  between $ \bf{x} $ and $ \bf{0} $, and $k \equiv i-j \; \text{mod} \;n $.

Thus, the minimum pairwise toroidal distance is equivalent to Eq. (\ref{minK})
\begin{align}
\label{minK}
    \min_{  i,j \in \{0,\cdots,n-1\} ,  i \ne j} \|{\bf{x}}_i - {\bf{x}}_j \|_{T_p} =   \min_{  k \in \{1,\cdots,n-1\} } \|{\bf{x}}_k \|_{T_p}.
\end{align}
Since the minimum value is smaller than the average value, it follows that
\begin{align}
\label{MeanIneq}
    \min_{  i,j \in \{0,\cdots,n-1\} ,  i \ne j} \|{\bf{x}}_i - {\bf{x}}_j \|_{T_p} =   \min_{  k \in \{1,\cdots,n-1\} } \|{\bf{x}}_k \|_{T_p} \le \frac{\sum_{k=1}^{n-1} {\| \boldsymbol{x}_k \|_{T_p} } }{n-1}.
\end{align}

Since  $n$ is a prime number, from  number theory, we know that for a primitive root $g$, the residue of $\{g^0,g^1,\cdots,g^{n-2} \} $  modulo $n$ forms a cyclic group under multiplication, and $g^{n-1} \equiv 1 \; \text{mod} \;n$.   Moreover, 
there is  a one-to-one correspondence between the residue of $\{g^0,g^1,\cdots,g^{n-2} \} $  modulo $n$ and the set $\{1,2,\cdots,n-1\}$. Then, for each $t^{th}$ component of $  \boldsymbol{z} = [ z_1,z_2, \cdots, z_d  ] $,  we know that $\exists m_t$ such that $ g^{m_t} \equiv z_t \; \text{mod} \; n$.  Therefore, the set $\left\{kz_t\; \text{mod} \; n \big| \forall k \in  \{1,\cdots,n-1\}  \right\}$  is a permutation of the set $\{1,\cdots,n-1\}$.

From the definition of the $l_p$-norm-based toroidal distance, we know that each $t^{th}$  component of $\|{\bf{x}}_k \|_{T_p}$ is determined by $\text{min}(kz_t\; \text{mod} \; n, n- kz_t\; \text{mod} \; n)$. Because the set $\left\{kz_t\; \text{mod} \; n \big| \forall k \in  \{1,\cdots,n-1\}  \right\}$   is a permutation of set $\{1,\cdots,n-1\}$, we know that the set $\left\{\text{min}(kz_t\; \text{mod} \; n, n- kz_t\; \text{mod} \; n) \big| \forall k \in  \{1,\cdots,n-1\}  \right\}$ 
 consists of two copy of permutation of the set $\{1,\cdots,\frac{n-1}{2}\}$. It follows that
 \begin{align}
     \sum_{k=1}^{n-1} {\| \boldsymbol{x}_k \|_{T_1} } = \frac{ \sum_{t=1}^d { \sum_{k=1}^{n-1} {\text{min}(kz_t\; \text{mod} \; n, n- kz_t\; \text{mod} \; n  )  }  }}{n}  = \frac{2d \sum_{k=1}^{\frac{n-1}{2}} {k}  }{n} = \frac{d(n+1)(n-1)}{4n}.
 \end{align}
 Similarly, for $l_2$-norm-based toroidal distance, we have that 
 \begin{align}
     \sum_{k=1}^{n-1} {\| \boldsymbol{x}_k \|^2_{T_2} } =  \frac{ \sum_{t=1}^d { \sum_{k=1}^{n-1} {\text{min}(kz_t\; \text{mod} \; n, n- kz_t\; \text{mod} \; n  )^2  }  }}{n^2}  = \frac{2d \sum_{k=1}^{\frac{n-1}{2}} {k^2}  }{n^2} = \frac{d(n-1)(n+1)}{12n}.
 \end{align}
 By Cauchy–Schwarz inequality, we know that 
 \begin{align}
       \sum_{k=1}^{n-1} {\| \boldsymbol{x}_k \|_{T_2} } \le \sqrt{ (n-1)  \sum_{k=1}^{n-1} {\| \boldsymbol{x}_k \|^2_{T_2} }  } = (n-1) \sqrt{\frac{d(n+1)}{12n}}.
 \end{align}
 
 Together with Eq.(\ref{MeanIneq}), it follows that
 \begin{align}
     & \min_{  i,j \in \{0,\cdots,n-1\} ,  i \ne j} \|{\bf{x}}_i - {\bf{x}}_j \|_{T_1} =  \min_{  k \in \{1,\cdots,n-1\} } \|{\bf{x}}_k \|_{T_1} \le \frac{(n+1)d}{4n}  \\
 &    \min_{  i,j \in \{0,\cdots,n-1\} ,  i \ne j} \|{\bf{x}}_i - {\bf{x}}_j \|_{T_2} =  \min_{  k \in \{1,\cdots,n-1\} } \|{\bf{x}}_k \|_{T_2} \le  \sqrt{\frac{(n+1)d}{12n} }. 
 \end{align}
 
 Now, we are going to prove the lower bound. For a non-degenerate rank-1 lattice, the components of generating vector $\boldsymbol{z}=[z_1,\cdots, z_d]$ should be all different.  Then, we know the components of $\boldsymbol{x}_k, \forall k \in \{1,\cdots, n-1\} $ should be all different. Thus,  the min norm point is achieved at $\boldsymbol{x}^* = [1/n,2/n,\cdots,d/n]$. Since $n\ge 2d+1$, it follows that 
 \begin{align}
      & \min_{  i,j \in \{0,\cdots,n-1\} ,  i \ne j} \|{\bf{x}}_i - {\bf{x}}_j \|_{T_1} =  \min_{  k \in \{1,\cdots,n-1\} } \|{\bf{x}}_k \|_{T_1} \ge \|{\bf{x}}^* \|_{T_1} = \frac{(d+1)d}{2n}  \\
 &    \min_{  i,j \in \{0,\cdots,n-1\} ,  i \ne j} \|{\bf{x}}_i - {\bf{x}}_j \|_{T_2} =  \min_{  k \in \{1,\cdots,n-1\} } \|{\bf{x}}_k \|_{T_2} \ge  \|{\bf{x}}^* \|_{T_2} =  \frac{\sqrt{6d(d+1)(2d+1)}}{6n} . 
 \end{align}

\end{proof}

\section{Proof of Corollary 1}
\label{proofc1}

\begin{corollary*}
\label{optimalCor}
Suppose $n=2d+1$ is a prime number. Let $g$ be a  primitive root of $n$. Let $  \boldsymbol{z} = [ g^0,g^{\frac{n-1}{2d}}, g^{\frac{2(n-1)}{2d}},\cdots, g^{\frac{(d-1)(n-1)}{2d} } ] \; \text{mod} \; n $. Construct rank-1 lattice $X=\{\boldsymbol{x}_0,\cdots,\boldsymbol{x}_{n-1}\}$ with $\boldsymbol{x}_i = \frac{i \boldsymbol{z}  \; \text{mod} \;n}{n}, i \in \{0,...,n-1\}  $. Then,  the pairwise toroidal distance of the lattice $X$  attains the upper bound. 
\begin{align}
    &  \|{\bf{x}}_i - {\bf{x}}_j \|_{T_1} = \frac{(n+1)d}{4n} ,  \forall i,j \in \{0,\cdots,n-1\} ,  i \ne j,  \\
 &  \|{\bf{x}}_i - {\bf{x}}_j \|_{T_2} = \sqrt{\frac{(n+1)d}{12n} } ,  \forall i,j \in \{0,\cdots,n-1\} ,  i \ne j . 
\end{align}

\end{corollary*}

\begin{proof}

From the definition of the rank-1 lattice, we know that 
\begin{align}
  \|\boldsymbol{x}_i - \boldsymbol{x}_j   \|_{T_p} =   \left \| \frac{i \boldsymbol{z}  \; \text{mod} \;n}{n} - \frac{j \boldsymbol{z}  \; \text{mod} \;n}{n} \right\|_{T_p} =   \left \| \frac{(i-j) \boldsymbol{z}  \; \text{mod} \;n}{n}   \right\|_{T_p} = \left \| \frac{k \boldsymbol{z}  \; \text{mod} \;n}{n}   \right\|_{T_p} =  \|\boldsymbol{x}_k  \|_{T_p}, 
\end{align}
where $ \|{\bf{x}} \|_{T_p} $ denote the $l_p$-norm-based toroidal distance, we know that  between $ \bf{x} $ and $ \bf{0} $, and $k \equiv i-j \; \text{mod} \;n $.

From Theorem 1, we know that $ \|\boldsymbol{x}_i - \boldsymbol{x}_j   \|_{T_p} \; \forall i,j \in \{0,\cdots,n-1\} ,  i \ne j$ has $\frac{n-1}{2d}$ different values.  Since $n=2d+1$, we know the pairwise toroidal distance has the same value. Therefore, we know that
\begin{align}
\label{CpairTp}
     \|\boldsymbol{x}_i - \boldsymbol{x}_j   \|_{T_p} =  \|\boldsymbol{x}_k  \|_{T_p}  = \frac{ \sum_{k=1}^{n-1} {\| \boldsymbol{x}_k \|_{T_p} }  }{n-1}   \;,  \forall i,j \in \{0,\cdots,n-1\} ,  i \ne j.
\end{align}

From  the proof of Theorem 2, we know that
\begin{align}
\label{CT1}
      \sum_{k=1}^{n-1} {\| \boldsymbol{x}_k \|_{T_1} } = \frac{ \sum_{t=1}^d { \sum_{k=1}^{n-1} {\text{min}(kz_t\; \text{mod} \; n, n- kz_t\; \text{mod} \; n  )  }  }}{n}  = \frac{2d \sum_{k=1}^{\frac{n-1}{2}} {k}  }{n} = \frac{d(n+1)(n-1)}{4n}.
\end{align}
and 
\begin{align}
\label{CT2}
     \sum_{k=1}^{n-1} {\| \boldsymbol{x}_k \|^2_{T_2} } =  \frac{ \sum_{t=1}^d { \sum_{k=1}^{n-1} {\text{min}(kz_t\; \text{mod} \; n, n- kz_t\; \text{mod} \; n  )^2  }  }}{n^2}  = \frac{2d \sum_{k=1}^{\frac{n-1}{2}} {k^2}  }{n^2} = \frac{d(n-1)(n+1)}{12n}.
 \end{align}
 
 Together  Eq.(\ref{CT1})  with Eq.(\ref{CpairTp}), we know that
 \begin{align}
      \|{\bf{x}}_i - {\bf{x}}_j \|_{T_1} = \frac{(n+1)d}{4n} ,  \forall i,j \in \{0,\cdots,n-1\} ,  i \ne j .
 \end{align}
 
 Since $\| \boldsymbol{x}_1  \|_{T_p} = \| \boldsymbol{x}_2  \|_{T_p} = \cdots = \| \boldsymbol{x}_{n-1}  \|_{T_p} $, it follows that
 \begin{align}
     \sum_{k=1}^{n-1} {\| \boldsymbol{x}_k \|_{T_2} }  = \sqrt{(n-1) \sum_{k=1}^{n-1} {\| \boldsymbol{x}_k \|^2_{T_2} } }  . 
 \end{align}
 
 Together with Eq.(\ref{CT2}), we know that
 \begin{align}
 \label{T2s}
     \sum_{k=1}^{n-1} {\| \boldsymbol{x}_k \|_{T_2} }  = \sqrt{(n-1) \sum_{k=1}^{n-1} {\| \boldsymbol{x}_k \|^2_{T_2} } }  = (n-1) \sqrt{\frac{d(n+1)}{12n}}.
 \end{align}
 
 Plug Eq.(\ref{T2s}) into Eq.(\ref{CpairTp}), if follows  that
 \begin{align}
      \|{\bf{x}}_i - {\bf{x}}_j \|_{T_2} = \sqrt{\frac{(n+1)d}{12n} } ,  \forall i,j \in \{0,\cdots,n-1\} ,  i \ne j . 
 \end{align}
 
 From Theorem 2, we know that the $l_1$-norm-based and $l_2$-norm-based pairwise toroidal distance of the lattice $X$  attains the upper bound.

\end{proof}

\section{Subgroup-based QMC on Sphere $\mathbb{S}^{d-1}$}
\label{sQMC}

In this section, we propose a closed-form subgroup-based QMC method  on the sphere $\mathbb{S}^{d-1}$ instead of unit cube $[0,1]^d$.
QMC uniformly on sphere can be used to construct samples for isotropic distribution, which is 
 helpful for variance reduction of the gradient estimators in Evolutionary strategy for reinforcement learning~\cite{salimans2017evolution}.

Lyu~\cite{lyu2017spherical} constructs structured sampling matrix on $\mathbb{S}^{d-1}$ by minimizing the  discrete Riesz energy. In contrast, we  construct samples by a closed-form construction without the time-consuming  optimization procedure. Our construction can achieve a small mutual coherence.

Without loss of generality, we assume that $d=2m, N=2n$, and $n$ is a prime such that $m|(n-1)$. Let $ F \in {\mathbb{C} ^{n \times n}} $ be a $n \times n$ discrete Fourier matrix. $ {F_{k,j}} = {e^{\frac{{2\pi ikj}}{n}}} $ is the $(k, j) ^{th} $entry of  $F$, where $i = \sqrt { - 1} $. Let $\Lambda  = \{{k_1},{k_2},...,{k_m} \} \subset \{ 1,...,n - 1\} $ be a subset of indexes.

The structured sampling matrix $\bf{V}$ in \cite{lyu2017spherical} can be defined as equation (\ref{eq22}).
\begin{equation}
\label{eq22}
\begin{array}{l}
\bf{V} = \frac{1}{{\sqrt m }}\left[ {\begin{array}{*{20}{c}}
{{\mathop{\rm Re}\nolimits} {F_\Lambda }}&{ - {\mathop{\rm Im}\nolimits} {F_\Lambda }}\\
{{\mathop{\rm Im}\nolimits} {F_\Lambda }}&{{\mathop{\rm Re}\nolimits} {F_\Lambda }}
\end{array}} \right] \in {\mathbb{R}^{d \times N}}
\end{array}
\end{equation}
where  $\text{Re}$ and $\text{Im}$ denote the real and image part of a complex number,   and $ {F_\Lambda } $ in equation (\ref{eq23}) is the matrix constructed by $m$ rows of $F$
\begin{equation}
\label{eq23}
\begin{array}{l}
{F_\Lambda }{\rm{ = }}\left[ {\begin{array}{*{20}{c}}
{{e^{\frac{{2\pi i{k_1}1}}{n}}}}& \cdots &{{e^{\frac{{2\pi i{k_1}n}}{n}}}}\\
 \vdots & \ddots & \vdots \\
{{e^{\frac{{2\pi i{k_m}1}}{n}}}}& \cdots &{{e^{\frac{{2\pi i{k_m}n}}{n}}}}
\end{array}} \right] \in {\mathbb{C}^{m \times n}}.
\end{array}
\end{equation}

With the $\bf{V}$ given in equation~(\ref{eq22}), we know that ${\left\| {{{\boldsymbol{v}}_i}} \right\|_2} = 1$ for $i \in \{1,...,n \}$. Thus, each column of matrix $\bf{V}$ is a point on $\mathbb{S}^{d-1}$.

Let $g$ denote a primitive root modulo $n$.  We construct the index $\Lambda  = \{{k_1},{k_2},...,{k_m} \}$ as 
\begin{equation}
    \label{Sindex}
   \Lambda =  \{ g^0,g^{\frac{n-1}{m}}, g^{\frac{2(n-1)}{m}},\cdots, g^{\frac{(m-1)(n-1)}{m} } \} \; \text{mod} \; n. 
\end{equation}

The set $ \{ g^0,g^{\frac{n-1}{m}}, g^{\frac{2(n-1)}{m}},\cdots, g^{\frac{(m-1)(n-1)}{m} } \} \; \text{mod} \; n $ forms a subgroup of the the group $\{g^0,g^1,\cdots,g^{n-2}\}$ mod $n$.   Based on this, we derive upper bounds of the mutual coherence of the points set $\boldsymbol{V}$.  The results are summarized in Theorem~\ref{Sbound1}
and Theorem~\ref{Sbound2}.

\begin{theorem}
\label{Sbound1}
Suppose $d=2m, N=2n$, and $n$ is a prime such that $m|(n-1)$. Construct matrix $\boldsymbol{V}$ as in Eq.(\ref{eq22}) with index set $\Lambda$  as Eq.(\ref{Sindex}). Let mutual coherence $\mu(\boldsymbol{V}):= \max_{i \ne j} |\boldsymbol{v}_i^\top \boldsymbol{v}_j| $. Then  $\mu(\boldsymbol{V}) \le \frac{\sqrt{n}}{m}$.
\end{theorem}

\begin{theorem}
\label{Sbound2}
Suppose $d=2m, N=2n$, and $n$ is a prime such that $m|(n-1)$, and $m \le n^\frac{2}{3}$. Construct matrix $\boldsymbol{V}$ as in Eq.(\ref{eq22}) with index set $\Lambda$  as Eq.(\ref{Sindex}). Let mutual coherence $\mu(\boldsymbol{V}):= \max_{i \ne j} |\boldsymbol{v}_i^\top \boldsymbol{v}_j| $. Then  $\mu(\boldsymbol{V}) \le C  m^{-1/2}n^{1/6} \log^{1/6}m$, where $C$ denotes a positive constant independent of $m$ and $n$.
\end{theorem}

Theorem~\ref{Sbound1} and Theorem~\ref{Sbound2} show that our construction can achieve a bounded mutual coherence. A smaller mutual coherence means that the points are more evenly spread on sphere $\mathbb{S}^{d-1}$.

\textbf{Remark:} Our construction does not require a restrictive constraint of the dimension of data. The only assumption of data dimension $d$ is that $d$ is a even number, i.e.,$2|d$, which is commonly satisfied in practice. Moreover, the product $\boldsymbol{V}^\top\boldsymbol{x}$ can be accelerated by fast Fourier transform as in \cite{lyu2017spherical}.

\subsection{Evaluation of the mutual coherence}

We evaluate our subgroup-based spherical QMC by comparing with the construction in \cite{lyu2017spherical} and i.i.d Gaussian sampling.

We set the dimension $d$ as in $\{50,100,200,500, 1000\}$. For each dimension $d$, we set the number of points $N=2n$, with $n$ as the first ten prime numbers such that $\frac{d}{2}$ divides $n\!-\!1$, i.e., $\frac{d}{2} \big | (n\!-\!1)$. Both subgroup-based QMC and Lyu's method are deterministic.  For Gaussian sampling method, we report the mean $\pm$ standard deviation of  mutual coherence over 50 independent runs. 
The mutual coherence for each dimension are reported in Table~\ref{tab:Coherence}.  The  smaller the mutual coherence, the better. 

We can observe that our subgroup-based spherical QMC achieves a competitive mutual coherence compared with Lyu's method in \cite{lyu2017spherical}. Note that our method does not require a time consuming optimization procedure, thus it is appealing for applications that demands a fast construction. Moreover, both our subgroup-based QMC and Lyu's method obtain a significant smaller coherence than i.i.d Gaussian sampling.


\begin{table}[t]
\centering
\begin{footnotesize}
\caption{Mutual coherence of points set constructed by different methods. Smaller is better.}
\label{tab:Coherence}
\resizebox{\columnwidth}{!}{
\begin{tabular}{ccccccccccccc}
\toprule
\multirow{4}{*}{\begin{tabular}[c]{@{}l@{}} d=50 \end{tabular}}
                          &          &   202    &     302   &      502    &     802    &    1202    &    1402   &     1502   &  2102   &     2302   &     2402
          \\  
&     SubGroup  & \cellcolor[HTML]{C0C0C0} \textbf{0.1490}   & \cellcolor[HTML]{C0C0C0} \textbf{0.2289}    & \cellcolor[HTML]{C0C0C0} \textbf{0.1923}   & \cellcolor[HTML]{C0C0C0} 0.2930  & \cellcolor[HTML]{C0C0C0} \textbf{0.2608}   & \cellcolor[HTML]{C0C0C0} 0.3402  & \cellcolor[HTML]{C0C0C0} 0.3358  & \cellcolor[HTML]{C0C0C0} 0.3211  & \cellcolor[HTML]{C0C0C0} 0.4534 & \cellcolor[HTML]{C0C0C0} \textbf{0.3353}   \\   
& Lyu~\cite{lyu2017spherical}      & 0.2313  &   0.2377  &   0.2901  &  \textbf{0.2902}    &  0.3005  &  \textbf{0.3154}    &  \textbf{0.3155}   &  \textbf{0.3209}   &  \textbf{0.3595}   & 0.3718     \\  

& \multirow{2}{*}{Gaussian} & \cellcolor[HTML]{C0C0C0} 0.5400$\pm$  & \cellcolor[HTML]{C0C0C0} 0.5738$\pm$ & \cellcolor[HTML]{C0C0C0} 0.5904$\pm$  & \cellcolor[HTML]{C0C0C0} 0.6158$\pm$  & \cellcolor[HTML]{C0C0C0} 0.6270$\pm$  & \cellcolor[HTML]{C0C0C0} 0.6254$\pm$  & \cellcolor[HTML]{C0C0C0} 0.6328$\pm$  & \cellcolor[HTML]{C0C0C0} 0.6447$\pm$ &  \cellcolor[HTML]{C0C0C0} 0.6520$\pm$ & \cellcolor[HTML]{C0C0C0} 0.6517$\pm$   \\ 
                                   &                   &      \cellcolor[HTML]{C0C0C0}0.0254  &  \cellcolor[HTML]{C0C0C0} 0.0291 &  \cellcolor[HTML]{C0C0C0} 0.0257  & \cellcolor[HTML]{C0C0C0} 0.0249   &  \cellcolor[HTML]{C0C0C0} 0.0209  & \cellcolor[HTML]{C0C0C0} 0.0184 &  \cellcolor[HTML]{C0C0C0} 0.0219  & \cellcolor[HTML]{C0C0C0} 0.0184  & \cellcolor[HTML]{C0C0C0} 0.0204  & \cellcolor[HTML]{C0C0C0} 0.0216 \\  
                            
\midrule

\multirow{4}{*}{\begin{tabular}[c]{@{}l@{}} d=100 \end{tabular}}
                  &     &  202    &     302   &      502     &    802    &    1202    &    1402     &   1502 &  2102     &   2302    &    2402      \\   
& SubGroup & \cellcolor[HTML]{C0C0C0} \textbf{0.1105}   & \cellcolor[HTML]{C0C0C0} \textbf{0.1529}    & \cellcolor[HTML]{C0C0C0} 0.1923  & \cellcolor[HTML]{C0C0C0} \textbf{0.1764}   & \cellcolor[HTML]{C0C0C0} 0.2397 & \cellcolor[HTML]{C0C0C0}  0.2749  & \cellcolor[HTML]{C0C0C0} 0.2513  & \cellcolor[HTML]{C0C0C0} 0.2679   & \cellcolor[HTML]{C0C0C0} 0.4534 & \cellcolor[HTML]{C0C0C0} 0.3353  \\  
& Lyu~\cite{lyu2017spherical}      &  0.1234  &  0.1581  &  \textbf{0.1586}   &  0.1870  &  \textbf{0.2041}   &   \textbf{0.2191}   &  \textbf{0.1976}   & \textbf{0.2047}   & \textbf{0.2244}  & \textbf{0.2218}    \\   

& \multirow{2}{*}{Gaussian} & \cellcolor[HTML]{C0C0C0} 0.4033$\pm$  & \cellcolor[HTML]{C0C0C0}  0.4210$\pm$  & \cellcolor[HTML]{C0C0C0} 0.4422$\pm$  & \cellcolor[HTML]{C0C0C0} 0.4577$\pm$  & \cellcolor[HTML]{C0C0C0} 0.4616$\pm$  & \cellcolor[HTML]{C0C0C0} 0.4734$\pm$  & \cellcolor[HTML]{C0C0C0} 0.4716$\pm$  & \cellcolor[HTML]{C0C0C0} 0.4878$\pm$  & \cellcolor[HTML]{C0C0C0} 0.4866$\pm$  & \cellcolor[HTML]{C0C0C0} 0.4947$\pm$   \\ 
                                   &                   &   \cellcolor[HTML]{C0C0C0}  0.0272 & \cellcolor[HTML]{C0C0C0}  0.0274  & \cellcolor[HTML]{C0C0C0} 0.0225   & \cellcolor[HTML]{C0C0C0} 0.0230  & \cellcolor[HTML]{C0C0C0} 0.0170  & \cellcolor[HTML]{C0C0C0} 0.0174 & \cellcolor[HTML]{C0C0C0}  0.0234  & \cellcolor[HTML]{C0C0C0} 0.0167  &  \cellcolor[HTML]{C0C0C0} 0.0172 & \cellcolor[HTML]{C0C0C0} 0.0192\\

\midrule

\multirow{4}{*}{\begin{tabular}[c]{@{}l@{}} d=200 \end{tabular}}
               &     &  202   &      802    &    1202    &    1402   &     2402    &    2602    &    3202 &  3602    &    3802    &    5602
        \\  
& SubGroup & \cellcolor[HTML]{C0C0C0} \textbf{0.0100}    & \cellcolor[HTML]{C0C0C0} 0.1251  & \cellcolor[HTML]{C0C0C0} 0.1835  & \cellcolor[HTML]{C0C0C0} 0.1966  & \cellcolor[HTML]{C0C0C0} 0.2365 & \cellcolor[HTML]{C0C0C0}  0.1553  & \cellcolor[HTML]{C0C0C0} 0.1910  & \cellcolor[HTML]{C0C0C0} 0.1914  & \cellcolor[HTML]{C0C0C0} 0.2529  & \cellcolor[HTML]{C0C0C0} 0.2457    \\  
& Lyu~\cite{lyu2017spherical}      &  \textbf{0.0100}   &  \textbf{0.1108}  &  \textbf{0.1223}  &  \textbf{0.1262}   & \textbf{0.1417}   & \textbf{0.1444}   & \textbf{0.1505}    & \textbf{0.1648}   & \textbf{0.1624}  & \textbf{0.1679}          \\ 

& \multirow{2}{*}{Gaussian} & \cellcolor[HTML]{C0C0C0} 0.2887$\pm$  & \cellcolor[HTML]{C0C0C0} 0.3295$\pm$  & \cellcolor[HTML]{C0C0C0}  0.3362$\pm$  & \cellcolor[HTML]{C0C0C0} 0.3447$\pm$ & \cellcolor[HTML]{C0C0C0}  0.3564$\pm$  & \cellcolor[HTML]{C0C0C0} 0.3578$\pm$  & \cellcolor[HTML]{C0C0C0} 0.3645$\pm$  & \cellcolor[HTML]{C0C0C0} 0.3648$\pm$  & \cellcolor[HTML]{C0C0C0} 0.3689$\pm$ & \cellcolor[HTML]{C0C0C0} 0.3768$\pm$  \\ 
                                   &                   &   \cellcolor[HTML]{C0C0C0}  0.0163  & \cellcolor[HTML]{C0C0C0}  0.0155  & \cellcolor[HTML]{C0C0C0} 0.0148  & \cellcolor[HTML]{C0C0C0} 0.0182  & \cellcolor[HTML]{C0C0C0} 0.0140  & \cellcolor[HTML]{C0C0C0} 0.0142 & \cellcolor[HTML]{C0C0C0}  0.0143  & \cellcolor[HTML]{C0C0C0} 0.0142  & \cellcolor[HTML]{C0C0C0} 0.0140  & \cellcolor[HTML]{C0C0C0} 0.0151\\

\midrule

\multirow{4}{*}{\begin{tabular}[c]{@{}l@{}} d=500 \end{tabular}}
            &     & 502    &    1502   &     4502   &     6002    &    6502     &   8002    &    9502 &  11002   &    14002   &    17002       \\  
& SubGroup & \cellcolor[HTML]{C0C0C0} \textbf{0.0040} &  \cellcolor[HTML]{C0C0C0} 0.0723  & \cellcolor[HTML]{C0C0C0} 0.1051 & \cellcolor[HTML]{C0C0C0}  0.1209  &  \cellcolor[HTML]{C0C0C0} 0.1107 &  \cellcolor[HTML]{C0C0C0} 0.1168  & \cellcolor[HTML]{C0C0C0} 0.1199  & \cellcolor[HTML]{C0C0C0} 0.1425  & \cellcolor[HTML]{C0C0C0} 0.1587  & \cellcolor[HTML]{C0C0C0} 0.1273  \\  
& Lyu~\cite{lyu2017spherical}      &  \textbf{0.0040}   & \textbf{0.0650}   & \textbf{0.0946}  & \textbf{0.0934}   & \textbf{0.0930}    & \textbf{0.1004}  &  \textbf{0.0980}    & \textbf{0.1022}   &  \textbf{0.1077}   & \textbf{0.1110}         \\ 

& \multirow{2}{*}{Gaussian} & \cellcolor[HTML]{C0C0C0} 0.2040$\pm$  & \cellcolor[HTML]{C0C0C0}  0.2218$\pm$  & \cellcolor[HTML]{C0C0C0} 0.2388$\pm$  & \cellcolor[HTML]{C0C0C0} 0.2425$\pm$  & \cellcolor[HTML]{C0C0C0} 0.2448$\pm$  & \cellcolor[HTML]{C0C0C0} 0.2498$\pm$  & \cellcolor[HTML]{C0C0C0} 0.2528$\pm$  & \cellcolor[HTML]{C0C0C0} 0.2527$\pm$ & \cellcolor[HTML]{C0C0C0}  0.2579$\pm$  & \cellcolor[HTML]{C0C0C0} 0.2607$\pm$  \\ 
                                   &                   &   \cellcolor[HTML]{C0C0C0}     0.0111   & \cellcolor[HTML]{C0C0C0} 0.0099   & \cellcolor[HTML]{C0C0C0} 0.0092   & \cellcolor[HTML]{C0C0C0} 0.0081   & \cellcolor[HTML]{C0C0C0} 0.0113   & \cellcolor[HTML]{C0C0C0} 0.0110   & \cellcolor[HTML]{C0C0C0} 0.0100   & \cellcolor[HTML]{C0C0C0} 0.0084  & \cellcolor[HTML]{C0C0C0} 0.0113 & \cellcolor[HTML]{C0C0C0} 0.0092\\
                                   
\midrule

\multirow{4}{*}{\begin{tabular}[c]{@{}l@{}} d=1000 \end{tabular}}
            &     & 6002    &    8002   &    11002    &   14002    &   17002    &   18002   &    21002   &    26002   &    32002   &    38002      \\  
& SubGroup & \cellcolor[HTML]{C0C0C0} 0.0754  & \cellcolor[HTML]{C0C0C0} 0.0778  & \cellcolor[HTML]{C0C0C0} 0.0819  & \cellcolor[HTML]{C0C0C0} 0.0921  & \cellcolor[HTML]{C0C0C0} 0.0935  & \cellcolor[HTML]{C0C0C0} 0.0764  & \cellcolor[HTML]{C0C0C0} 0.1065  & \cellcolor[HTML]{C0C0C0} 0.0931 & \cellcolor[HTML]{C0C0C0}  0.0908  & \cellcolor[HTML]{C0C0C0} 0.1125  \\   
& Lyu~\cite{lyu2017spherical}      & \textbf{0.0594}    &  \textbf{0.0637}   & \textbf{0.0662}  &  \textbf{0.0680}     &  \textbf{0.0684}   &  \textbf{0.0744}   & \textbf{0.0774}   &  \textbf{0.0815}   & \textbf{0.0781}   & \textbf{0.0814}          \\  

& \multirow{2}{*}{Gaussian} & \cellcolor[HTML]{C0C0C0} 0.1736$\pm$ & \cellcolor[HTML]{C0C0C0}  0.1764$\pm$  & \cellcolor[HTML]{C0C0C0} 0.1797$\pm$  & \cellcolor[HTML]{C0C0C0} 0.1828$\pm$  & \cellcolor[HTML]{C0C0C0} 0.1846$\pm$  & \cellcolor[HTML]{C0C0C0} 0.1840$\pm$  & \cellcolor[HTML]{C0C0C0} 0.1869$\pm$  & \cellcolor[HTML]{C0C0C0} 0.1888$\pm$  & \cellcolor[HTML]{C0C0C0} 0.1909$\pm$  & \cellcolor[HTML]{C0C0C0} 0.1920$\pm$  \\ 
                                   &                   &   \cellcolor[HTML]{C0C0C0}   \cellcolor[HTML]{C0C0C0} 0.0067  & \cellcolor[HTML]{C0C0C0}  0.0059 & \cellcolor[HTML]{C0C0C0}  0.0060  & \cellcolor[HTML]{C0C0C0} 0.0062  & \cellcolor[HTML]{C0C0C0} 0.0052 & \cellcolor[HTML]{C0C0C0} 0.0057  & \cellcolor[HTML]{C0C0C0} 0.0052   & \cellcolor[HTML]{C0C0C0} 0.0055 & \cellcolor[HTML]{C0C0C0}  0.0067  & \cellcolor[HTML]{C0C0C0} 0.0056 \\

 \bottomrule
\end{tabular}
}
\end{footnotesize}
\vspace{-2mm}
\end{table}
\vspace{-2mm}


\newpage

\section{Proof of Theorem~\ref{Sbound1}}
\label{proofT3}

\begin{proof}
Let $\boldsymbol{c}_i \in \mathbb{C}^m$ be the $i^{th}$ column of matrix $F_{\Lambda} \in \mathbb{C}^{m \times n}$ in Eq.(\ref{eq23}). Let $\boldsymbol{v}_i \in \mathbb{R}^{2m}$ be the  $i^{th}$ column of matrix $\boldsymbol{V} \in \mathbb{R}^{2m \times 2n}$ in Eq.(\ref{eq22}).
For $1 \le i,j \le n$, $i \ne j$, we know that 
\begin{align}
 & \boldsymbol{v}_i^\top \boldsymbol{v}_{i+n} = 0, \\
    &\boldsymbol{v}_{i+n}^\top \boldsymbol{v}_{j+n} = \boldsymbol{v}_i^\top \boldsymbol{v}_j = \text{Re} (\boldsymbol{c}_i^*\boldsymbol{c}_j), \\
    & \boldsymbol{v}_{i+n}^\top \boldsymbol{v}_{j} = -\boldsymbol{v}_i^\top \boldsymbol{v}_{j+n} =  \text{Im} (\boldsymbol{c}_i^*\boldsymbol{c}_j),
\end{align}
where $*$ denote the complex conjugate, $\text{Re}(\cdot)$ and $\text{Im}(\cdot)$ denote the real and image part of the input complex number.

It follows that
\begin{align}
\mu(V) \le=\max_{1 \le k,r \le 2n, k \ne r}    |\boldsymbol{v}_k^\top \boldsymbol{v}_r| \le  \max_{1 \le i,j \le n, i \ne j}    |\boldsymbol{c}_i^*\boldsymbol{v}_j| = \mu(F_\Lambda) 
\end{align}
From the definition of $F_\Lambda$ in Eq.(\ref{eq23}), we know that 
\begin{align}
    \mu(F_\Lambda) & = \max_{1 \le i,j \le n, i \ne j}    |\boldsymbol{c}_i^*\boldsymbol{v}_j| = 
     \max_{1 \le i,j \le n, i \ne j}    { \frac{1}{m} \left| \sum_{z \in \Lambda}{e^{\frac{{2\pi \boldsymbol{i}{z(j-i)}}}{n}}} \right|}  \\
    & = \max_{ 1 \le k \le n-1} { \frac{1}{m} \left| \sum_{z \in \Lambda}{e^{\frac{{2\pi \boldsymbol{i}{zk}}}{n}}} \right|} 
\end{align}
Because $\Lambda$ is a subgroup of the multiplicative group $\{g^0,g^1,\cdots,g^{n-2}\}$ mod $n$, from \cite{bourgain2006estimates}, we know that 
\begin{align}
     \max_{ 1 \le k \le n-1} {  \left| \sum_{z \in \Lambda}{e^{\frac{{2\pi \boldsymbol{i}{zk}}}{n}}} \right|}  \le \sqrt{n}
\end{align}

Finally, we know that 
\begin{align}
  \mu(V) \le   \mu(F_\Lambda) \le \frac{\sqrt{n}}{m}.
\end{align}

\end{proof}

\section{Proof of Theorem~\ref{Sbound2}}
\label{proofT4}

\begin{proof}
Because $\Lambda$ is a subgroup of the multiplicative group $\{g^0,g^1,\cdots,g^{n-2}\}$ mod $n$, and $m \le n^{2/3}$,
from Theorem 1 in \cite{shkredov2014exponential}, we know that
\begin{align}
    \max_{ 1 \le k \le n-1} {  \left| \sum_{z \in \Lambda}{e^{\frac{{2\pi \boldsymbol{i}{zk}}}{n}}} \right|} \le C  m^{1/2}n^{1/6} \log^{1/6}m
\end{align}
From the proof of Theorem~\ref{Sbound1}, we have that
\begin{align}
   \mu(V) \le   \mu(F_\Lambda) = \max_{ 1 \le k \le n-1} { \frac{1}{m} \left| \sum_{z \in \Lambda}{e^{\frac{{2\pi \boldsymbol{i}{zk}}}{n}}} \right|}  \le C  m^{-1/2}n^{1/6} \log^{1/6}m
\end{align}

\end{proof}



\section{QMC for  Generative models}
\label{QMCgm}

Our subgroup rank-1 lattice can be used for generative models. Buchholz et al.~\cite{buchholz2018quasi} suggest using QMC for variational inference to maximize the evidence lower bound (ELBO). We present a new method by directly learning the inverse of the cumulative distribution function (CDF).

In variational autoencoder, the objective is the evidence lower bound (ELBO)~\cite{kingma2013auto} defined as
\begin{align}
 \mathcal{L}(x,\phi,\theta)     = \mathbb{E}_{q_\phi(z|x)} \left[ \log p_{\theta}(x|z) \right] - \text{KL}\left[ q_\phi(z|x) || p_\theta(z) \right]. 
\end{align}

The ELBO consists of two terms, i.e., the reconstruction term $\mathbb{E}_{q_\phi(z|x)} \left[ \log p_{\theta}(x|z) \right]$ and the regularization term $\text{KL}\left[ q_\phi(z|x) || p_\theta(z) \right]$. The reconstruction term is learning to fit, while  the regularization term controls the distance between distribution $q_\phi(z|x)$ to the prior distribution $ p_\theta(z)$.

The reconstruction term $\mathbb{E}_{q_\phi(z|x)} \left[ \log p_{\theta}(x|z) \right]$ can be reformulated as  \begin{align}
   \mathbb{E}_{q_\phi(z|x)} \left[ \log p_{\theta}(x|z) \right] &  =   \int _{\mathcal{Z}} q_\phi(z|x) \log p_{\theta}(x|z)   \text{d}\boldsymbol{z} \\ & = \int _{[0,1]^d}  \log p_{\theta}\left(x| \Phi^{-1} (\boldsymbol{\epsilon}) \right)  \text{d}\boldsymbol{\epsilon}.  
   \label{I}
\end{align}
where $\Phi^{-1}(\cdot)$ denotes  the  inverse cumulative distribution function with respect to the density $q_\phi(z|x)$.

 Eq.(\ref{I}) provides an alternative training scheme, 
we  directly learn the  inverse of  CDF  $F(\boldsymbol{\epsilon};x) = \Phi^{-1}(\boldsymbol{\epsilon})$ given $x$ instead of the density $q_\phi(z|x)$. We  parameterize $F(\boldsymbol{\epsilon},x)$ as a neural network with input $\boldsymbol{\epsilon}$ and data $x$. The inverse of CDF function $F(\boldsymbol{\epsilon},x)$ can be seen as an encoder of $x$ for inference. It is worth  noting that learning the inverse of CDF can bring more flexibility without the assumption of the distribution, e.g., Gaussian.

To ensure the distribution $q$ close to the prior distribution $p(z)$, we can use other  regularization terms instead of the KL-divergence for any implicit distribution $q$, e.g.,  the maximum mean discrepancy. Besides this, we can also use a discriminator-based adversarial loss    similar to adversarial autoencoders~\cite{makhzani2015adversarial}
\begin{equation}
\widetilde{L}(x,F,D)\! =\! \mathbb{E}_{p_\theta(\boldsymbol{z})}\left[ \log (D(\boldsymbol{z})) \right] \! + \!\mathbb{E}_{p_(\boldsymbol{\epsilon})}\left[ \log (1\!-\!D(F(\boldsymbol{\epsilon},x))) \right], 
\end{equation}
where $p(\boldsymbol{\epsilon})$ denotes a uniform distribution on unit cube $[0,1]^d$, $D$ is the discriminator, $F$ denotes the inverse of CDF mapping.

When the domain $\mathcal{Z}$ coincides with  a target domain $\mathcal{Y}$, we  can use an empirical data distribution $Y$ as the prior.  This leads to a training scheme similar to  cycle GAN~\cite{zhu2017unpaired}.  In contrast to cycle GAN, the encoder $F$ depends on both data $x$ in source domain and $\boldsymbol{\epsilon}$ in unit cube.  The expectation term $\mathbb{E}_{p_(\boldsymbol{\epsilon})}[\cdot]$ can be approximated by QMC methods. 

\section{Generative Inference for CycleGAN}
\label{GenerativeExp}

We evaluate our subgroup rank-1 lattice on training generative model. As shown in section~\ref{QMCgm}, we can learn the inverse CDF  functions $F(\boldsymbol{\epsilon},x)$ as a generator from domain $\mathcal{X}$ to domain $\mathcal{Y}$ in cycle GAN. We set $F(\boldsymbol{\epsilon},x)= G_1(x) + G_2(\boldsymbol{\epsilon})$, where $G_1$ and $G_2$ denotes the  neural networks. Network $G_1$  maps input image $x$ to a target mean, while network $G_2$ maps $\boldsymbol{\epsilon} \in [0,1]^d$ as the residue. Similarly, $\widetilde{F}(\boldsymbol{\widetilde{\epsilon}},y) =  \widetilde{G}_1(y) + \widetilde{G}_2(\boldsymbol{\widetilde{\epsilon}})$ denotes an generator from domain $\mathcal{Y}$ to domain $\mathcal{X}$.  

We implement the model  based on the open-sourced Pytorch code of \cite{zhu2017unpaired}.
All $G_1$, $G_2$, $\widetilde{G}_1$ and $\widetilde{G}_2$ employ the default ResNet architecture  with 9 blocks in \cite{zhu2017unpaired}. The input size of both $\boldsymbol{\epsilon}$ and $\boldsymbol{\widetilde{\epsilon}}$ are $d=256 \times 256$.  We keep all the hyperparameters same for all the methods as the default value in \cite{zhu2017unpaired}.

We compare our subgroup rank-1 lattice with Monte Carlo sampling for training the generative model. For subgroup rank-1 lattice, we set the number of points $n=12d+1=786433$. We do not store all the points,  instead  we sample $i\in \{0,\cdots, n-1\} $ uniformly and construct $\boldsymbol{\epsilon}$ and $\boldsymbol{\widetilde{\epsilon}}$ based on Eq.(\ref{rank1lattice}) during the training process. For Monte Carlo sampling, $\boldsymbol{\epsilon}$ and $\boldsymbol{\widetilde{\epsilon}}$ are sampled from  $Uniform [0,1]^d$.


We train generative models on the Vangogh2photo data set and maps data set employed in \cite{zhu2017unpaired}.
We present  experimental results of the generated images from models trained with subgroup-based rank-1 lattice sampling, Monte-Carlo sampling, and standard version of CycleGAN. The experimental results on Vangogh2photo dataset and maps dataset are shown in Figure~\ref{vango1} and Figure~\ref{maps1}, respectively. From Figure~\ref{vango1}, we can observe that the images generated by the model trained with Monte-Carlo sampling have some blurred patches.  This phenomenon  may be because the additional flexibility of randomness makes the training more difficult to converge to a good model. In contrast, the model trained with subgroup-based rank-1 lattice sampling generates more clearer images.  It may be because the rank-1 lattice sampling has finite possible choices, i.e., $n=786433$ possible points in the experiments,  which is much smaller than the case of Monte-Carlo uniform sampling. The rank-1 lattice sampling is more deterministic than Monte Carlo sampling, which alleviates the training difficulty to fit a good model.   Since in our subgroup-based rank-1 lattice it is very simple to construct new samples, it can serve as a good alternative to Monte Carlo sampling for generative model training.

\begin{figure*}[t]
\centering
{
\includegraphics[width=1\linewidth]{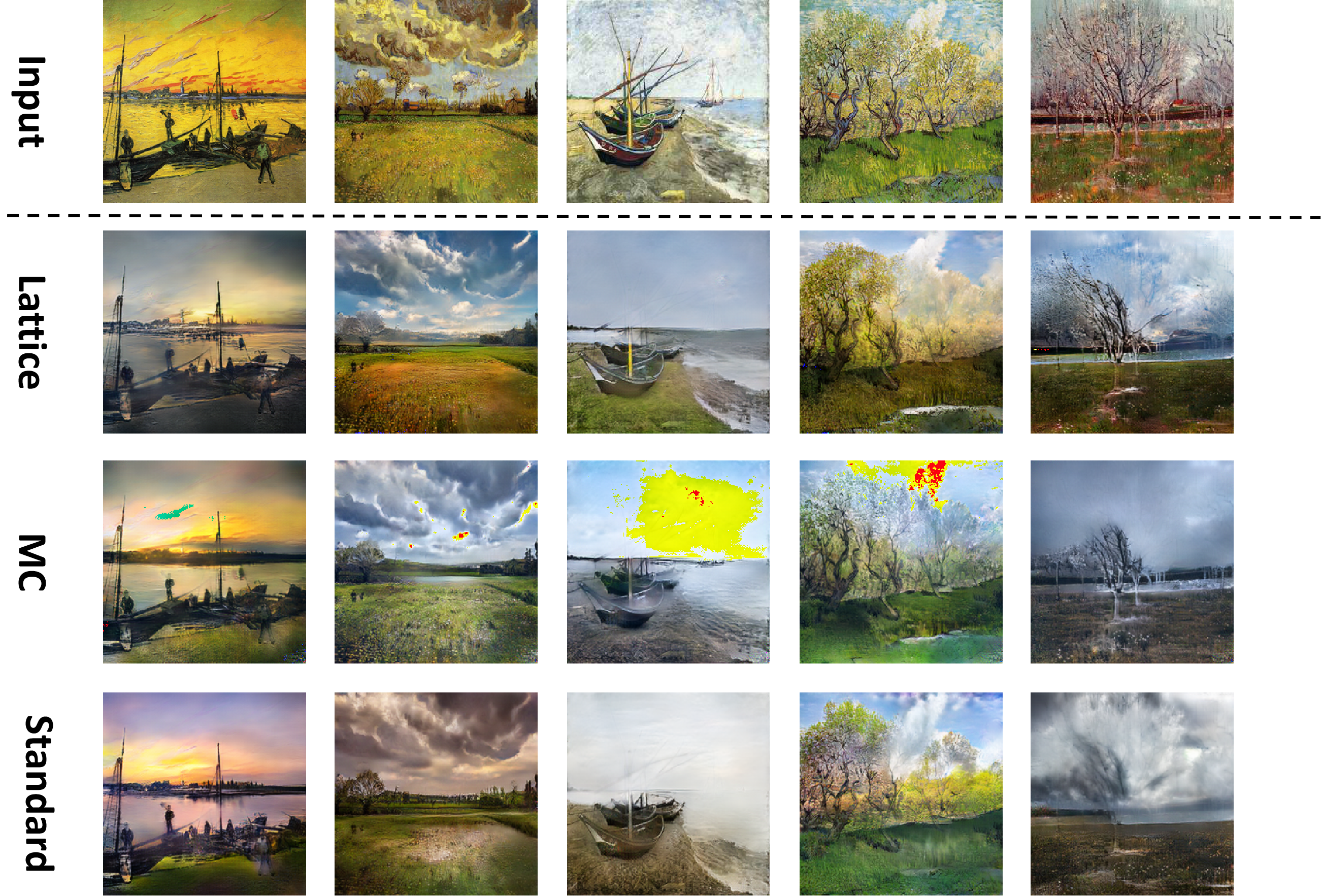}}
\caption{Illustration of the generated images from models trained with subgroup rank-1 lattice sampling, Monte-Carlo sampling, and Standard version of CycleGAN.}
\label{vango1}
\end{figure*}

\begin{figure*}[t]
\centering
{
\includegraphics[width=1\linewidth]{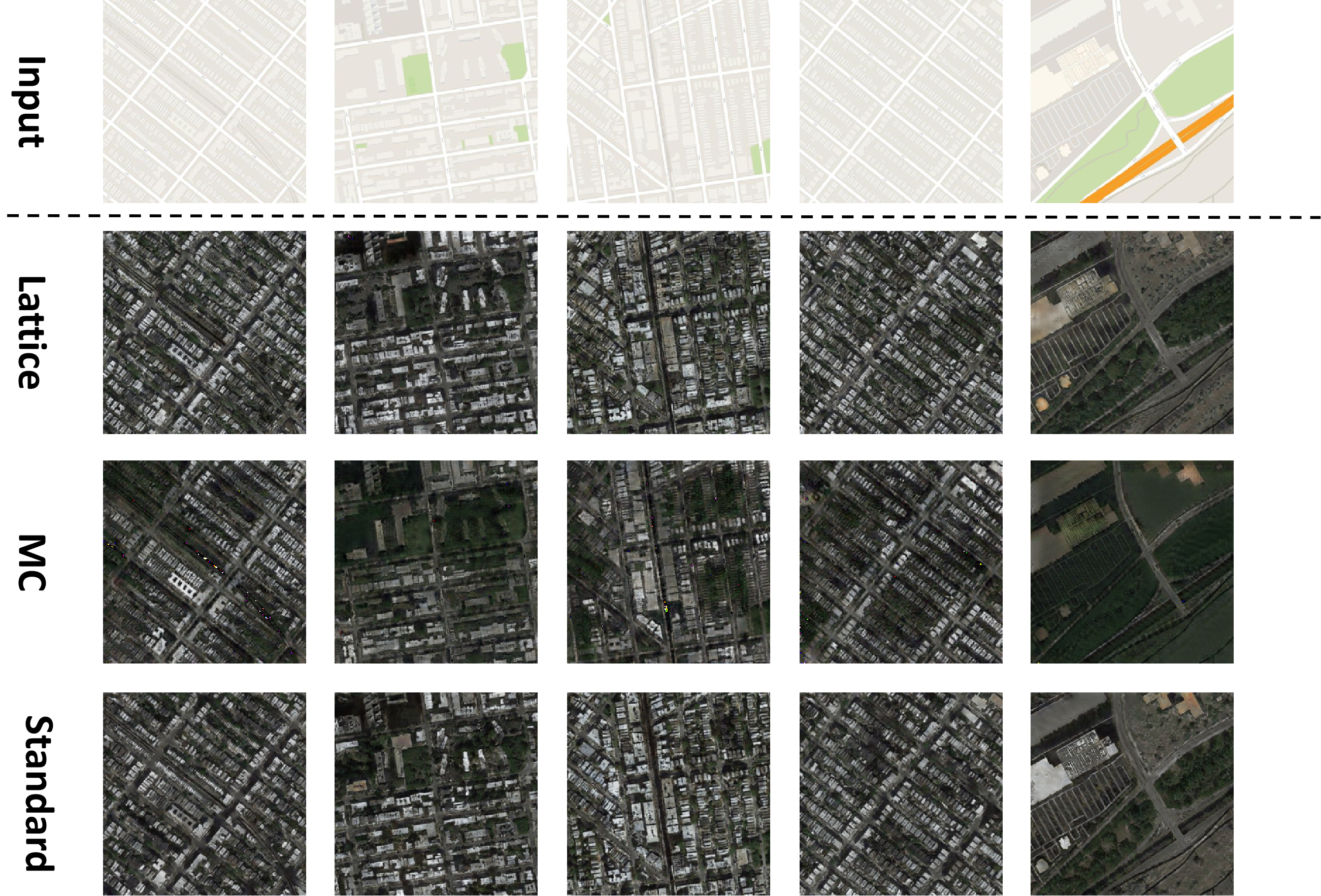}}
\caption{Illustration of the generated images from models trained with subgroup rank-1 lattice sampling, Monte-Carlo sampling, and Standard version of CycleGAN.}
\label{maps1}
\end{figure*}

\begin{figure}[t]
\centering
\subfigure[\scriptsize{$\frac{{{{\left\| {\widetilde K - K} \right\|}_F}}}{{{{\left\| K \right\|}_F}}}$ for Gaussian Kernel }]{
\label{fig2a_K_meanG_SIFT1M}
\includegraphics[width=0.3\linewidth]{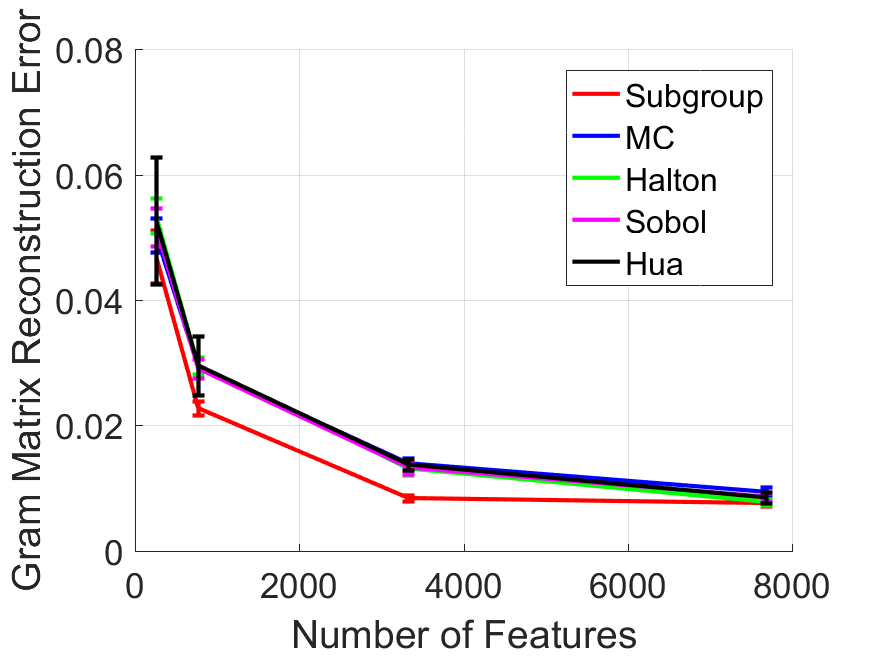}}
\subfigure[\scriptsize{$\frac{{{{\left\| {\widetilde K - K} \right\|}_F}}}{{{{\left\| K \right\|}_F}}}$ for First-order Arc Kernel}]{
\label{fig2c_K_acrCosine_SIFT1M}
\includegraphics[width=0.3\linewidth]{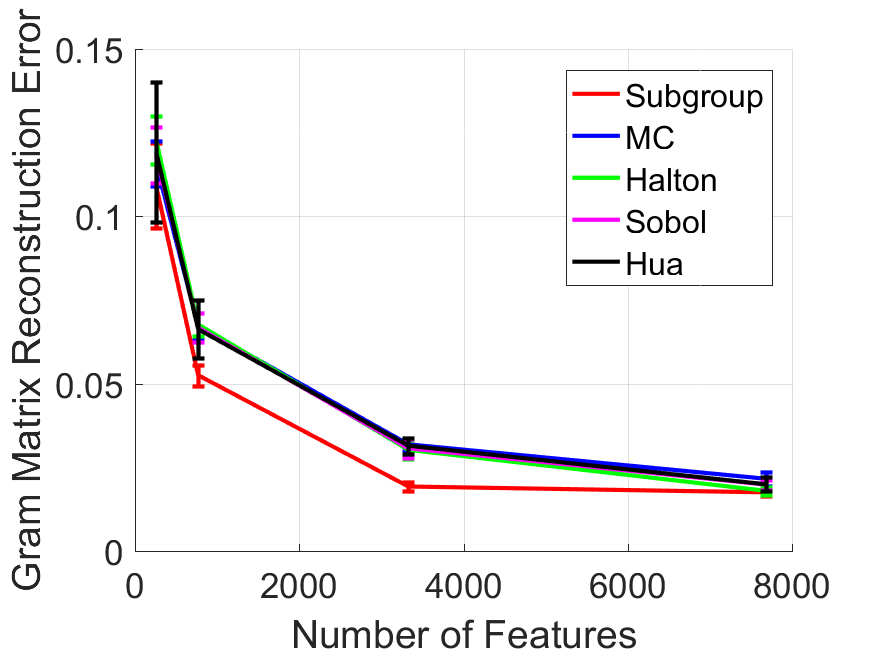}}
\subfigure[\scriptsize{$\frac{{{{\left\| {\widetilde K - K} \right\|}_F}}}{{{{\left\| K \right\|}_F}}}$ for  Zero-order Arc Kernel}]{
\label{fig2c_K_meanAngle_SIFT1M}
\includegraphics[width=0.3\linewidth]{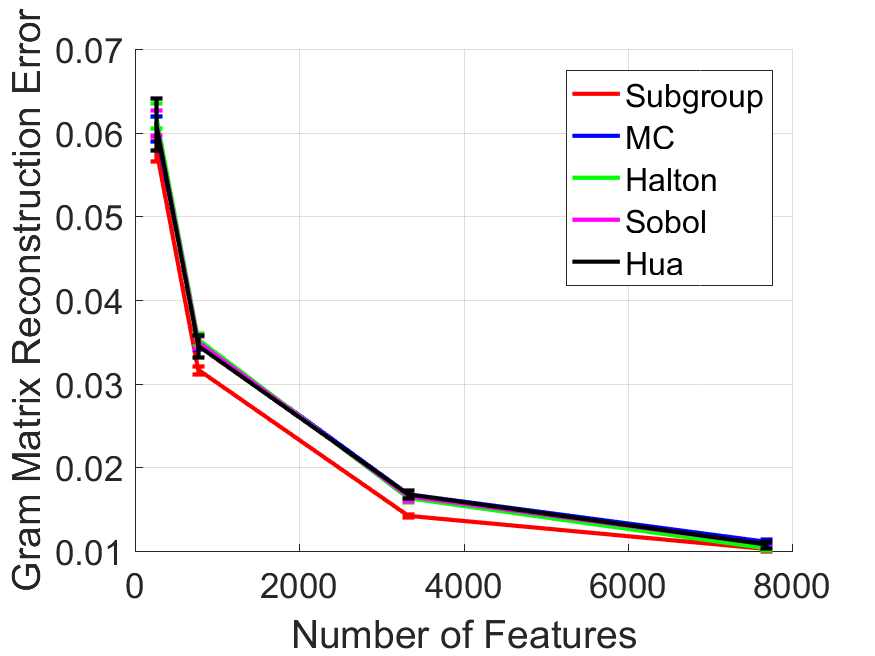}}
\subfigure[\scriptsize{$\frac{{{{\left\| {\widetilde K - K} \right\|}_\infty }}}{{{{\left\| K \right\|}_\infty }}}$ for Gaussian Kernel}]{
\label{fig2c_K_Max_Gaussian_SIFT1M}
\includegraphics[width=0.3\linewidth]{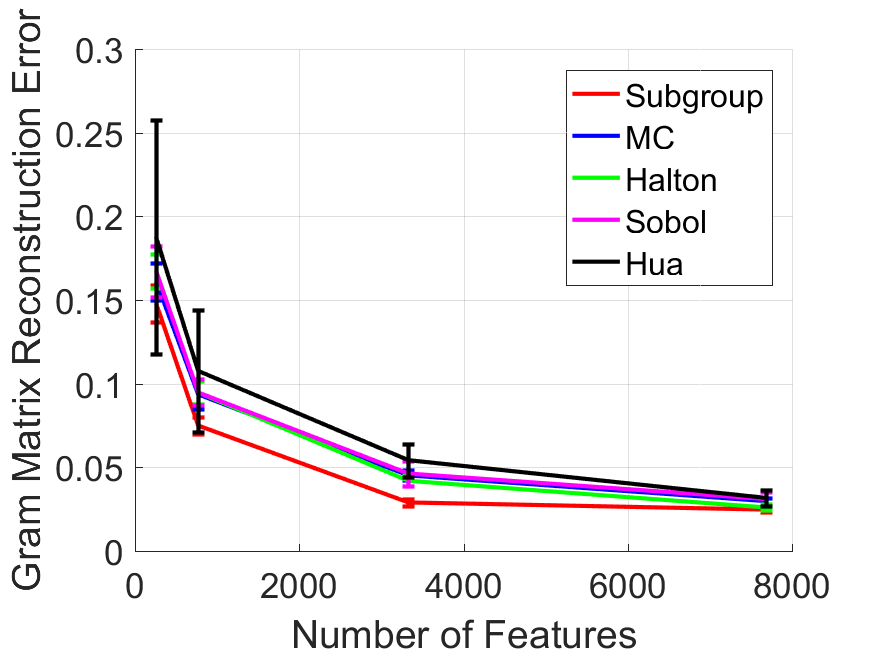}}
\subfigure[\scriptsize{$\frac{{{{\left\| {\widetilde K - K} \right\|}_\infty }}}{{{{\left\| K \right\|}_\infty }}}$  for First-order Arc Kernel}]{
\label{fig2c_K_Max_acrCosine_SIFT1M}
\includegraphics[width=0.3\linewidth]{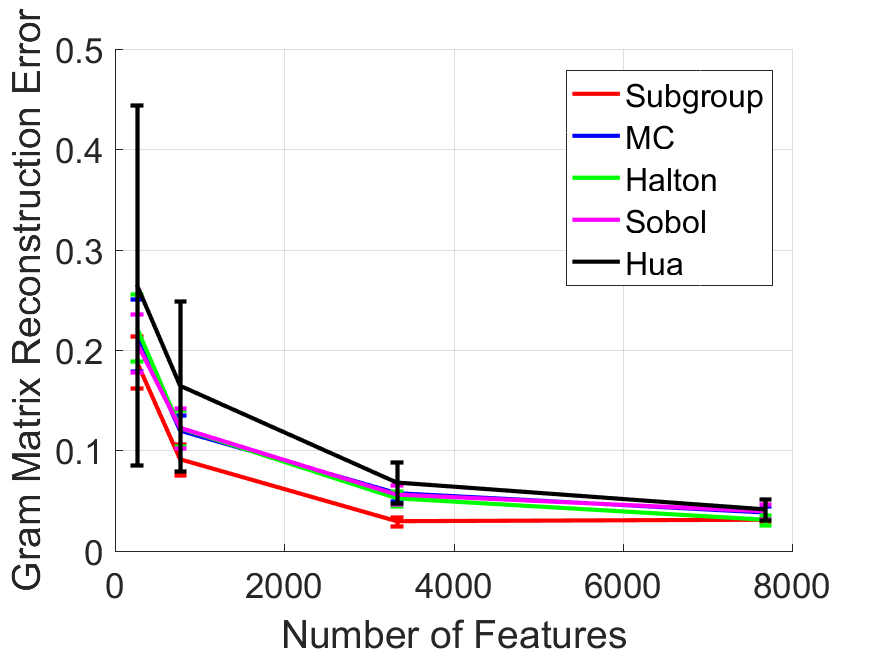}}
\subfigure[\scriptsize{$\frac{{{{\left\| {\widetilde K - K} \right\|}_\infty }}}{{{{\left\| K \right\|}_\infty }}}$  for Zero-order Arc Kernel}]{
\label{fig2c_K_maxAngle_SIFT1M}
\includegraphics[width=0.3\linewidth]{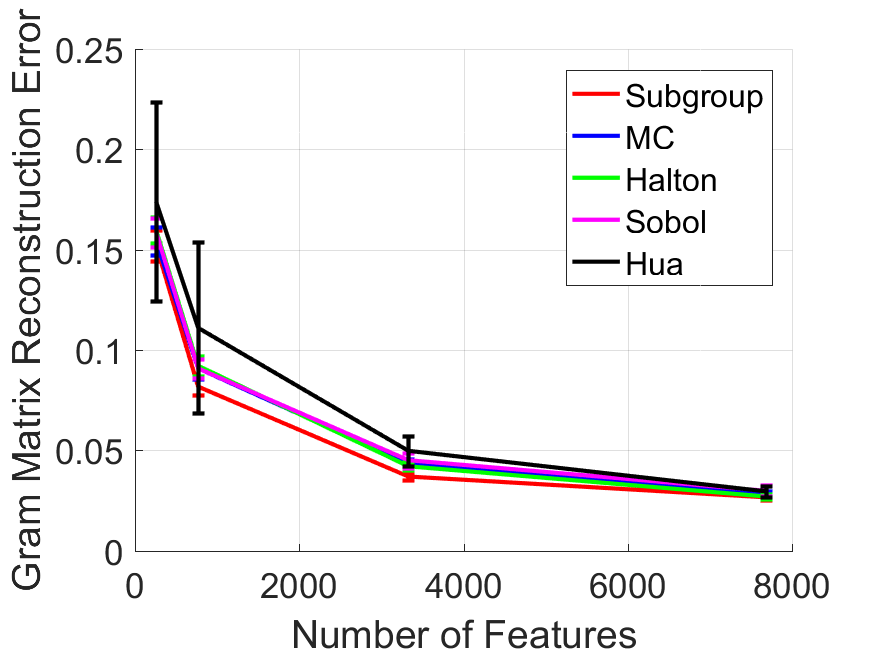}}
\caption{Relative Mean and Max Reconstruction Error for Gaussian, Zero-order and First-order Arc-cosine Kernel on SIFT1M dataset.}
\label{fig_KernelApp2}
\end{figure}

\end{document}